\documentclass[11pt,a4paper]{article}
\usepackage[utf8]{inputenc}
\usepackage[dvipsnames]{xcolor}
\usepackage[T1]{fontenc}
\usepackage{amsmath,amsthm,enumerate}
\usepackage{amssymb}
\usepackage{vmargin}
\usepackage{authblk}
\usepackage{subfig}
\usepackage{tikz}
\usepackage{adjustbox}
\usepackage{comment}

\usepackage{geometry}
\usepackage{longtable}
\usepackage{array} 
\newcolumntype{L}{>{\centering\arraybackslash}m{2.4cm}}

\usepackage{hyperref}
\usepackage{authblk}

\newtheorem{definition}{Definition}

\newtheorem{theorem}{Theorem}

\newtheorem{proposition}{Proposition}
\newtheorem{corollary}{Corollary}
\newtheorem{remark}{Remark}
\newtheorem{lemma}{Lemma}

\newtheorem*{notation}{Notation}
\newtheorem{problem}{Problem}

\DeclareMathOperator*{\argmax}{arg\!\max}

\usepackage{algpseudocode} \usepackage{algorithm}

\title{The Maximum Cover with Rotating Field of View}

\author[1]{\href{mailto:potapov@liverpool.ac.uk}{Igor Potapov}}
\author[1]{\href{mailto:jfralph@liverpool.ac.uk}{Jason F. Ralph}}
\author[1]{\href{mailto:Theofilos.Triommatis@liverpool.ac.uk}{Theofilos Triommatis} }

 \affil[1]{
School of Electrical Engineering and Electronics and Computer Science, \newline
University of Liverpool,
Liverpool, L69-3BX, UK }

\providecommand{\keywords}[1]
{
  \small	
  \textbf{\textit{Keywords---}} #1
}

\usepackage{url}

\usepackage{pgf,tikz,pgfplots}
\pgfplotsset{compat=1.15}
\usepackage{mathrsfs}
\usetikzlibrary{arrows}
\usepackage{amssymb,fancyhdr,txfonts,pxfonts}

\begin{document}

\maketitle

\begin{abstract}
    Imagine a polygon-shaped platform $P$ and only one static spotlight outside $P$; which direction should the spotlight face to light most of $P$? This problem occurs in maximising the visibility, as well as in limiting the uncertainty in localisation problems. 
    More formally, we define the following maximum cover problem: ``Given a convex polygon $P$ and a Field Of View (FOV) with a given centre and inner angle $\phi$; find the direction (an angle of rotation $\theta$)
    of the FOV such that the intersection between the FOV and $P$ has the maximum area''. In this paper, we provide the theoretical foundation for the analysis of the maximum cover with a rotating field of view. The main challenge is that the function of the area  $A_{\phi}(\theta)$, with the angle of rotation $\theta$ and the fixed inner angle $\phi$, cannot be approximated directly. We found an alternative way to express it by various compositions of a function $A_{\theta}(\phi)$ (with a restricted inner angle $\phi$ and a fixed direction $\theta$). We show that $A_{\theta}(\phi)$    
    that has an analytical solution in the special case of a two-sector intersection and later provides a constrictive solution for the original problem. Since the optimal solution is a real number, we develop an algorithm that approximates the direction of the field of view, with precision $\varepsilon$, and complexity $\mathcal{O}(n(\log{n}+(\log{\varepsilon})/\phi))$. 
\end{abstract}


\section{Introduction}
The use of antennas, sensors and cameras in ``smart'' or autonomous systems motivates the study of various visibility problems~\cite{CZYZOWICZ201316,ERDEM2006156,Sen_acc_points_21,MARENGONI2000773} 
with applications in computer graphics, motion planning, and other areas.
The most known visibility problems are the art gallery problem, region visibility, point or edge visibility, viewshed, see~\cite{9719824,Asano1986,BAREQUET2014407,Visibility2002,visibility2007,Survey2003,10.1007/978-3-030-58150-3_25}. 
Point or edge visibility is the decision problem of checking whether these objects are visible from a viewpoint in a context of a given set of obstacles.
In the art gallery problem, the objective is to find the minimal number of locations to place guards (with restricted or unrestricted Field of View, FOV)  within a polygon room to observe the room's whole area~\cite{AGP2021,restricted-FOV}.

In this paper, we study the 
problem of  finding
%
the maximal visibility area from a viewpoint
with a rotating FOV.
Imagine a polygon-shaped platform $P$ and only one static spotlight outside of $P$. Which direction should the spotlight face to light most of $P$?
More formally, we define the following problem: ``Given a polygon $P$ and a Field Of View (FOV) with a given centre and inner angle $\phi$; find the direction (as an angle $\theta$) of the FOV such that the intersection between the FOV and $P$ has the maximum area''. 
This problem occurs in maximising the visibility, as well as in limiting the uncertainty in localisation problems. The occurrence in the former is straightforward to understand. However, the occurrence in the latter is more subtle because we assume inside the polygon an object which we need to detect by the maximising probability of detection in the following scan without prior knowledge of its position. 
%
In~\cite{FUSION2022}, the geometric approach for passive localisation of static emitters is based on the problem of finding the maximum intersection of a polygon and a rotating FOV. For a passive sensor, a measurement is an angle with an error that points to the direction of a transmission's origin point. The angle with its angular error creates a cone of possible locations for the emitter. After consecutive iterations, a sensor computes a polygon by intersecting multiple measurements from different positions. A sensor needs to make a decision to move to its next position from a given finite set. The choice is made by evaluating all the available positions according to an objective function.
In a myopic (greedy) decision-making strategy, a sensor moves by minimising the maximum uncertainty on its subsequent measurement, achieved by evaluating the maximal intersection of polygons that contain the emitters' position and FOVs with centres that represent sensors' available positions, see Figure~\ref{fig:Geometric_Exploration}.
Experimental results in  \cite{FUSION2022} were based on a heuristic to estimate the intersection. Here we provide an algorithm with proven guarantee and precision.
\begin{figure}[ht]
    \centering
    \begin{tabular}{c c c c}
         \scalebox{0.49}{\includegraphics{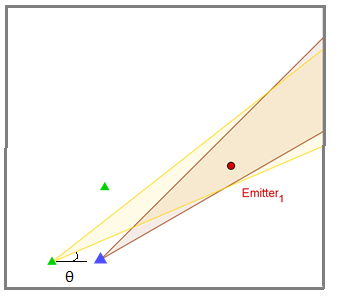}} 
         &
         \scalebox{0.5}{\includegraphics{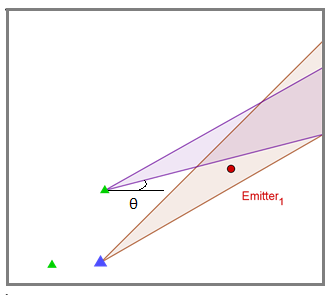}} 
         & 
         \scalebox{0.3}{
        \includegraphics{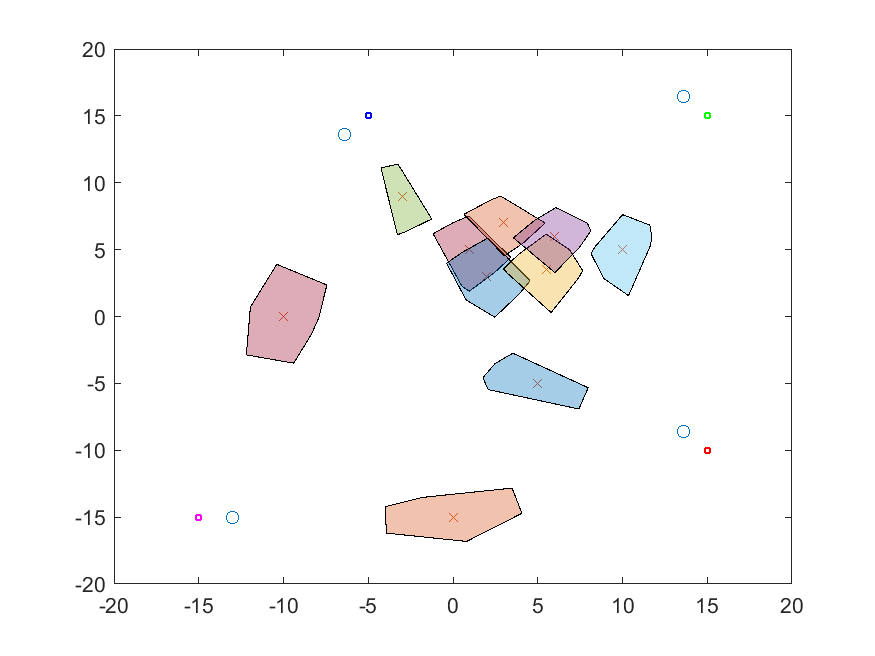}}
        &
        \scalebox{0.3}{
        \includegraphics{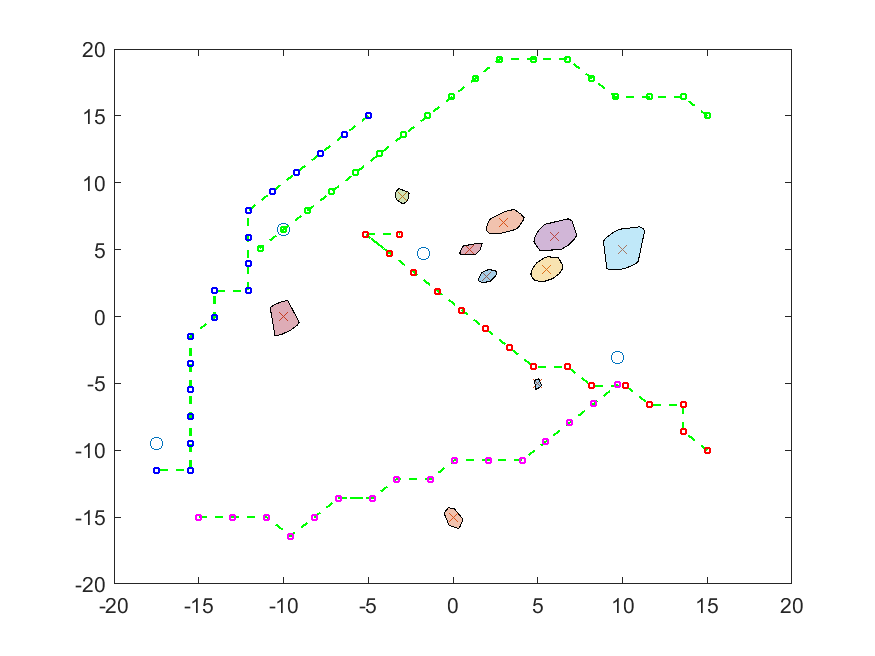}} 
         \\
         (a) & (b) & (c) & (d)
    \end{tabular}
    \caption{In (a) and (b), the sensor (blue triangle) calculates the worst-case uncertainty, which corresponds to the maximum cover, for two positions of ``East'' and ``North'' (green triangles) to select the one which minimises it, which in this case is ``North''. In (c) and (d), the sensors explore by minimising the maximum uncertainty in each step which makes the areas of uncertainty (the polygons) smaller.}
    \label{fig:Geometric_Exploration}
\end{figure}

There are also several related problems in the literature. One example is finding the intersection between two static polyhedra in the three-dimensional space, which has a 
linear-time algorithm on the number of vertices~\cite{DBLP:journals/dcg/Chan16}. In \cite{translation1998}, authors allow some flexibility and aim to compute the maximum overlap of two convex polygons under translations. The problem of approximating the intersection in the general case under the operation of translation has been recently solved in~\cite{translation2017}. The closest formulation to our problem is the Maximum Cover under Rotation (MCR):
Given a set of finite points $S$, a point $r$ on the plane, compute an angle $\theta \in [0,2\pi)$  such that, after counterclockwise rotation of a polygon $P$ by $\theta$ around $r$, the number of points of $S$ contained in $P$ is maximized. The problem is 3SUM-hard, but it has efficient solutions with respect to the number of points in $S$ and vertices in $P$~\cite{MCR}.
%
 
However, the problem we study is quite different to the one mentioned above. On the one hand, we consider a polygon essentially an infinite set of points instead of a finite one, but on the other hand, the Field of View is a cone in 2D, a specific shape, and the centre of rotations is its vertex. One might assume that expressing the area of the intersection as a function of rotations would be enough to provide an approximation through the use of a numerical method. Unfortunately, a naive application of numerical methods to find the maximum of $A_{\phi}(\theta)$, with the angle of rotation $\theta$ and the fixed inner angle $\phi$,
would not guarantee the maximum as we do not know the number and distribution of its extreme points. 

In this paper, we design an algorithm with a mathematical guarantee and provide the theoretical foundation for analysis of the maximum cover with a rotating field of view. 
We show an alternative way to express the maximum cover by various compositions of a function $A_{\theta}(\phi)$ (with a variable inner angle $\phi$ and a fixed direction $\theta$) that has an analytical solution.
The core component of the solution is to find the maximal intersection of a fixed sector (field of view with infinite radius) and a rotated one under a restricted rotation angle \footnote{ We consider to be restricted, that is the domain of the angle of rotations to be a closed interval, a proper subset of $[0,2\pi]$ since the area of intersection of two sectors without restrictions can be infinite.}. Surprisingly, the function of the area, even in such a restricted case, is non-monotonic. Nonetheless, it is possible to find the maximal intersection as shown in Section~\ref{sec:Max_Area} by using functions that calculate the area with a fixed rotation angle and the inner angle as a variable.
Later, we show how to express more complex shapes of the intersection of a polygon and a rotated sector as a combination of multiple two-sector intersections. 
 %
%
%
Finally, we complete the solution by identifying how an infinite number of intersections can be decomposed into a finite number of equivalence classes and propose at the same time a partitioning algorithm as well as a solution for each equivalence class. Moreover, our solution can be directly applied to special cases of non-convex polygons.
Since the optimal solution is a real-value number, we develop an algorithm that approximates the direction of the field of view, with precision $\varepsilon$, and complexity $\mathcal{O}(n(\log{n}+(\log{\varepsilon})/\phi))$.  


\section{The Maximum Intersection Problem}
\label{sec:Problem_Analysis}
To begin with we introduce the notations we will use throughout the paper.
Let $P$ and $Q$ be two points, we will notate with $\hat{PQ}$ the slope's angle of the line that $P$ and $Q$ define. In other words the slope of the line that $P$ and $Q$ define is $\tan{ \left( {\hat{PQ}} \right)}$. 
Throughout this paper, when we mention angles we mean the positive (counterclockwise) angles and we will use the notation $\hat{ABC}$ to denote the positive angle with apex $B$. 
Moreover, denote a convex polygon $\mathcal{P}=(P_1,\ldots,P_n)$ as the list of its vertices $P_i \in \mathbb{R}^2$, $i \in \{1,\ldots,n\}$, in counter-clockwise order. A field of view in the 3D space is in essence a cone, and we assume that its height tends to infinity. Since we study the problem on a 2D plane, the field of view is actually a sector of a circle with a radius that tends to infinity. So we formulate the sector in the following way: 

\begin{definition}
    A \textbf{sector} $S[C,\varepsilon_r,\varepsilon_{\ell}]$ is the set of points that lie inside an angle $0<\hat{RCL}<\pi$ that is formed by two half lines $\varepsilon_r$, and $\varepsilon_{\ell}$, $R \in \varepsilon_r$ and $L \in \varepsilon_{\ell}$, that share a common endpoint $C$, called the centre of the sector. We will call $\varepsilon_r$, and $\varepsilon_{\ell}$ the \textbf{right} and the \textbf{left semi-line} of the sector respectively.
\end{definition}
As we are interested in studying the sector under rotation we introduce an alternative definition that is based on the angles of the arrays' slopes. 

\begin{definition}
    A \textbf{sector} $S[C,\varepsilon_r,\varepsilon_{\ell}]$ defined by two semi-lines $\varepsilon_r,\varepsilon_{\ell}$ with gradients  $\theta$, $\theta+\phi$ and a common endpoint $C$ can be represented by another triplet $S(C,\theta,\phi)$, where the angle $\phi$ is the \textbf{inner angle of the sector} and the angle $\theta$ is the \textbf{direction of the sector}.     
\end{definition}

Note that $\varepsilon_r$ is a half line that extends from $C$, and the direction $\theta \in [0,2\pi]$ corresponds to exactly one semi-line because if $x$ is the horizontal line that passes through $C$, then $\theta=\hat{XCR}$ where $X\in x$ and $R\in \varepsilon_r$. Now we are ready to formulate the problem properly.

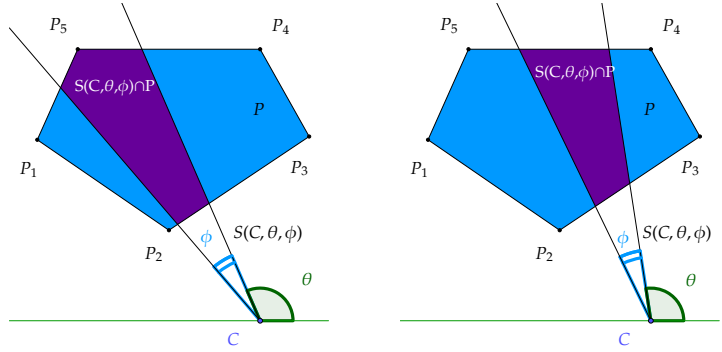
\begin{figure}[ht]
    \centering
    \begin{tabular}{c c}
        \scalebox{0.6}{
\definecolor{qqzzqq}{rgb}{0.,0.6,0.}
\definecolor{qqwuqq}{rgb}{0.,0.39215686274509803,0.}
\definecolor{ffffff}{rgb}{1.,1.,1.}
\definecolor{wwqqzz}{rgb}{0.4,0.,0.6}
\definecolor{zzttqq}{rgb}{0.6,0.2,0.}
\definecolor{qqzzff}{rgb}{0.,0.6,1.}
\definecolor{ududff}{rgb}{0.30196078431372547,0.30196078431372547,1.}
\begin{tikzpicture}[line cap=round,line join=round,>=triangle 45,x=1.0cm,y=1.0cm]
\clip(2.5,3.) rectangle (9.5,11.);
\draw [shift={(8.,4.)},line width=2.pt,color=qqzzff] (0,0) -- (113.36066987521403:1.562641319656736) arc (113.36066987521403:130.360669875214:1.562641319656736) -- cycle;
\fill[line width=2.pt,color=zzttqq,fill=zzttqq,fill opacity=0.10000000149011612] (3.6231006819848326,9.149997586671411) -- (4.,10.) -- (5.408456604181931,10.) -- (6.880104882191272,6.592806556006475) -- (6.190932838771259,6.1286053979143045) -- cycle;
\fill[line width=2.pt,color=qqzzff,fill=qqzzff,fill opacity=1.0] (4.,10.) -- (8.,10.) -- (9.07479033035953,8.071066645652698) -- (6.,6.) -- (3.112712679976908,7.998944738591618) -- cycle;
\fill[line width=2.pt,color=wwqqzz,fill=wwqqzz,fill opacity=1.0] (3.6231006819848326,9.149997586671411) -- (4.,10.) -- (5.408456604181931,10.) -- (6.880104882191272,6.592806556006475) -- (6.190932838771259,6.1286053979143045) -- cycle;
\draw [shift={(8.,4.)},line width=2.pt,color=qqwuqq,fill=qqwuqq,fill opacity=0.10000000149011612] (0,0) -- (0.:0.7212190706108013) arc (0.:113.36066987521403:0.7212190706108013) -- cycle;
\draw [line width=0.4pt] (3.112712679976908,7.998944738591618)-- (6.,6.);
\draw [line width=0.4pt] (6.,6.)-- (9.07479033035953,8.071066645652698);
\draw [line width=0.4pt] (9.07479033035953,8.071066645652698)-- (8.,10.);
\draw [line width=0.4pt] (8.,10.)-- (4.,10.);
\draw [line width=0.4pt] (4.,10.)-- (3.112712679976908,7.998944738591618);
\draw [line width=0.4pt,domain=2.5:8.0] plot(\x,{(-20.191718252942543--2.075693883629488*\x)/-0.8965417959766597});
\draw [shift={(8.,4.)},line width=2.pt,color=qqzzff] (113.36066987521403:1.562641319656736) arc (113.36066987521403:130.360669875214:1.562641319656736);
\draw [shift={(8.,4.)},line width=2.pt,color=qqzzff] (113.36066987521403:1.4063771876910622) arc (113.36066987521403:130.360669875214:1.4063771876910622);
\draw [line width=0.4pt,domain=2.5:8.0] plot(\x,{(-19.639945207352795--1.7228724795933834*\x)/-1.464241342651432});
\draw [color=ffffff](3.808378743145003,9.453301608439161) node[anchor=north west] {\small{S(C,$\theta$,$\phi$)$\cap$P}};
\draw [line width=0.4pt,color=qqzzqq,domain=2.5:9.5] plot(\x,{(--8.-0.*\x)/2.});
\begin{normalsize}
\draw [fill=black] (3.112712679976908,7.998944738591618) circle (1.0pt);
\draw[color=black] (2.9324079123242073,7.349847575041899) node {$P_1$};
\draw [fill=black] (4.,10.) circle (1.0pt);
\draw[color=black] (3.5815050758739284,10.523211485729414) node {$P_5$};
\draw [fill=black] (6.,6.) circle (1.0pt);
\draw[color=black] (5.6730403806452525,5.474677991453823) node {$P_2$};
\draw [fill=black] (9.07479033035953,8.071066645652698) circle (1.0pt);
\draw[color=black] (8.870444927019802,7.373888210728926) node {$P_3$};
\draw[color=black] (7.970444927019802,8.673888210728926) node {$P$};
\draw [fill=black] (8.,10.) circle (1.0pt);
\draw[color=black] (8.461754120340348,10.499170850042386) node {$P_4$};
\draw [fill=ududff] (8.,4.) circle (1.5pt);
\draw[color=ududff] (7.4159864679546885,3.587488090022233) node {$C$};
\draw[color=qqzzff] (6.828514731886035,5.823267208915708) node {$\phi$};
\draw[color=qqwuqq] (9.026709058985478,4.933763688495723) node {$\theta$};
\draw[color=black] (8.17479033035953,5.9197230528086975) node {$S(C,\theta,\phi)$};
\draw[color=qqzzqq] (-3.234015141398143,4.404869703381138) node {$n$};
\end{normalsize}
\end{tikzpicture}
        } 
        &
        \scalebox{.6}{
\definecolor{qqzzqq}{rgb}{0.,0.6,0.}
\definecolor{qqwuqq}{rgb}{0.,0.39215686274509803,0.}
\definecolor{ffffff}{rgb}{1.,1.,1.}
\definecolor{wwqqzz}{rgb}{0.4,0.,0.6}
\definecolor{zzttqq}{rgb}{0.6,0.2,0.}
\definecolor{qqzzff}{rgb}{0.,0.6,1.}
\definecolor{ududff}{rgb}{0.30196078431372547,0.30196078431372547,1.}
\begin{tikzpicture}[line cap=round,line join=round,>=triangle 45,x=1.0cm,y=1.0cm]
\clip(2.5,3.) rectangle (9.5,11.);
\draw [shift={(8.,4.)},line width=2.pt,color=qqzzff] (0,0) -- (98.69920321100707:1.562641319656736) arc (98.69920321100707:115.69920321100709:1.562641319656736) -- cycle;
\fill[line width=2.pt,color=zzttqq,fill=zzttqq,fill opacity=0.10000000149011612] (7.081956375264599,10.) -- (7.5357145366587,7.03440131273794) -- (6.783519344358747,6.527750059671657) -- (5.112497770133332,10.) -- cycle;
\fill[line width=2.pt,color=qqzzff,fill=qqzzff,fill opacity=1.0] (4.,10.) -- (8.,10.) -- (9.07479033035953,8.071066645652698) -- (6.,6.) -- (3.112712679976908,7.998944738591618) -- cycle;
\fill[line width=2.pt,color=wwqqzz,fill=wwqqzz,fill opacity=1.0] (5.112497770133332,10.) -- (7.081956375264599,10.) -- (7.5357145366587,7.03440131273794) -- (6.783519344358747,6.527750059671657) -- cycle;
\draw [shift={(8.,4.)},line width=2.pt,color=qqwuqq,fill=qqwuqq,fill opacity=0.10000000149011612] (0,0) -- (0.:0.7212190706108013) arc (0.:98.69920321100707:0.7212190706108013) -- cycle;
\draw [line width=0.4pt] (3.112712679976908,7.998944738591618)-- (6.,6.);
\draw [line width=0.4pt] (6.,6.)-- (9.07479033035953,8.071066645652698);
\draw [line width=0.4pt] (9.07479033035953,8.071066645652698)-- (8.,10.);
\draw [line width=0.4pt] (8.,10.)-- (4.,10.);
\draw [line width=0.4pt] (4.,10.)-- (3.112712679976908,7.998944738591618);
\draw [line width=0.4pt,domain=2.5:8.0] plot(\x,{(-16.633704171262597--1.9314500695073287*\x)/-0.29552590380099186});
\draw [shift={(8.,4.)},line width=2.pt,color=qqzzff] (98.69920321100707:1.562641319656736) arc (98.69920321100707:115.69920321100709:1.562641319656736);
\draw [shift={(8.,4.)},line width=2.pt,color=qqzzff] (98.69920321100707:1.4063771876910622) arc (98.69920321100707:115.69920321100709:1.4063771876910622);
\draw [line width=0.4pt,domain=2.5:8.0] plot(\x,{(-17.47446850762324--1.7606514750909703*\x)/-0.8473141767238692});
\draw [color=ffffff](5.336979427114712,9.777951779431588) node[anchor=north west] {\small{S(C,$\theta$,$\phi$)$\cap$P}};
\draw [line width=0.4pt,color=qqzzqq,domain=2.5:9.5] plot(\x,{(--8.-0.*\x)/2.});
\begin{normalsize}
\draw [fill=black] (3.112712679976908,7.998944738591618) circle (1.0pt);
\draw[color=black] (2.9324079123242073,7.349847575041899) node {$P_1$};
\draw [fill=black] (4.,10.) circle (1.0pt);
\draw[color=black] (3.5815050758739284,10.523211485729414) node {$P_5$};
\draw [fill=black] (6.,6.) circle (1.0pt);
\draw[color=black] (5.6730403806452525,5.474677991453823) node {$P_2$};
\draw [fill=black] (9.07479033035953,8.071066645652698) circle (1.0pt);
\draw[color=black] (8.870444927019802,7.373888210728926) node {$P_3$};
\draw[color=black] (7.970444927019802,8.673888210728926) node {$P$};
\draw [fill=black] (8.,10.) circle (1.0pt);
\draw[color=black] (8.461754120340348,10.499170850042386) node {$P_4$};
\draw [fill=ududff] (8.,4.) circle (1.5pt);
\draw[color=ududff] (7.4159864679546885,3.587488090022233) node {$C$};
\draw[color=qqzzff] (7.396799181695222,5.799226573228682) node {$\phi$};
\draw[color=qqwuqq] (8.87479033035953,4.9097230528086975) node {$\theta$};
\draw[color=black] (8.57479033035953,5.9197230528086975) node {$S(C,\theta,\phi)$};
\draw[color=qqzzqq] (-3.234015141398143,4.404869703381138) node {$n$};
\end{normalsize}
\end{tikzpicture}
        } 
    \end{tabular}
    \caption{
    The intersection between a polygon $P$ and a sector $S(C,\theta,\phi)$ with centre $C$, inner angle $\phi$ and direction $\theta$. The intersection becomes a quadrilateral from a pentagon as  $S$ is rotated clockwise.}
    \label{fig:Problem1}
\end{figure}

\begin{problem}
\label{prob:Max_Inter_Prob}
 Given a convex polygon $\mathcal{P} = (P_1,\ldots, P_n)$, a point $C=(x_0,y_0)$ outside of the polygon on the Euclidean plane and $0<\phi<\pi$, find the direction $\theta \in [0,2\pi]$ such that the intersection $S(C,\theta,\phi)\cap \mathcal{P}$ has the maximum area.
\end{problem}

    Let $\mathcal{P}$ be a convex set, and a sector $S(C,\theta,\phi)$ with its centre $C$ outside of $\mathcal{P}$.
    We will say that the sector $S(C,\theta,\phi)$ \textbf{contains} $\mathcal{P}$ if $\mathcal{P}\cap S(C,\theta,\phi) = \mathcal{P}$; \textbf{fully intersects} $\mathcal{P}$ if both semi-lines of $S(C,\theta,\phi)$ intersect an edge or a vertex of $\mathcal{P}$; \textbf{partially intersects} $\mathcal{P}$ if only one of the two semi-lines of $S(C,\theta,\phi)$ intersects an edge or a vertex of $\mathcal{P}$; \textbf{does not intersect} $\mathcal{P}$ if none of the two semi-lines of $S(C,\theta,\phi)$ intersect an edge or a vertex of $\mathcal{P}$.


\section{Studying the Area of Intersection}
\label{sec:Max_Area}
As Figure~\ref{fig:Local_extreme_points_problem} presents, it is intuitive to think that if a sector is rotated towards a ``corner'', then the area of intersection should decrease. In other words, in many cases, it is easy to assume that the area of intersection as a function of rotations is monotonic. But this is not the case, as there are examples where the function has local extreme points, a crucial fact especially when the domain of the rotations is restricted (is a bounded interval). 
In this section, we study some fundamental cases under restricted rotations to extract the formulae of the area of the respective intersection. 
We show that the function of the area of the intersection of a rotating sector and a static one is $A(\theta,\phi)$, which depends on two values - the direction of the rotating field of view $\theta$ and its inner angle $\phi$.
The straightforward approach would be to consider the function $A_{\phi}(\theta)$, where $\phi$ is constant and $\theta$ is variable. However, maximising the area through function $A_{\phi}(\theta)$ leads to the analysis of polynomials of trigonometric functions with rational exponents. The direct maximisation of this non-convex function is difficult as there are no constructive criteria to check the number of possible solutions that would guarantee finding the maximum value. Instead, we found a more elegant way to solve the problem by expressing the function $A_{\phi}(\theta)$  by a composition of $A_{\theta}(\phi)$ functions with an inner angle
$\phi$ and a fixed direction $\theta$. The key is that function $A_{\theta}(\phi)$ has two local extreme points calculated analytically and by expressing $A_{\phi}(\theta)$  as a composition of $A_{\theta}(\phi)$ functions allows us to identify the intervals with only one solution in each one, where the application of classical numerical algorithms yields the maximum.
Finally, we prove that the function of any intersection's area is expressed as $A(\theta,\phi)$ or as a linear combination of $A_\theta(\phi)$ functions.

\begin{figure}[ht]
    \centering
    \begin{tabular}{c c c}
        \scalebox{0.4}{
        \definecolor{ududff}{rgb}{0.30196078431372547,0.30196078431372547,1.}
        \definecolor{uuuuuu}{rgb}{0.26666666666666666,0.26666666666666666,0.26666666666666666}
        \definecolor{qqwuqq}{rgb}{0.,0.39215686274509803,0.}
        \definecolor{cqcqcq}{rgb}{0.7529411764705882,0.7529411764705882,0.7529411764705882}
        \begin{tikzpicture}[line cap=round,line join=round,>=triangle 45,x=1.0cm,y=1.0cm]
        \draw [color=cqcqcq,, xstep=2.0cm,ystep=2.0cm] (-13.,4.) grid (-0.5,14.);
        \clip(-13.,4.) rectangle (-0.5,14.);
        \draw [shift={(-5.192322082401852,4.745049387656246)},line width=2.pt,color=qqwuqq,fill=qqwuqq,fill opacity=0.10000000149011612] (0,0) -- (0.027829690692931185:0.5373931719486121) arc (0.027829690692931185:82.0198862548315:0.5373931719486121) -- cycle;
        \fill[line width=2.pt,color=ududff,fill=ududff,fill opacity=0.10000000149011612] (-5.016394426465813,6.) -- (-3.9328371234625554,13.729375428089941) -- (-6.141596458369784,12.454147631506153) -- (-5.346852796853905,6.) -- cycle;
        \draw [line width=2.pt,domain=-13.:-0.5] plot(\x,{(--16.--0.5773502691896257*\x)/1.});
        \draw [line width=2.pt,domain=-13.:-0.5] plot(\x,{(--6.-0.*\x)/1.});
        \draw [line width=2.pt,color=ududff] (-5.016394426465813,6.)-- (-3.9328371234625554,13.729375428089941);
        \draw [line width=2.pt,color=ududff] (-3.9328371234625554,13.729375428089941)-- (-6.141596458369784,12.454147631506153);
        \draw [line width=2.pt,color=ududff] (-6.141596458369784,12.454147631506153)-- (-5.346852796853905,6.);
        \draw [line width=2.pt,color=ududff] (-5.346852796853905,6.)-- (-5.016394426465813,6.);
        \draw [line width=2.pt,domain=-5.192322082401852:-0.5] plot(\x,{(--52.62583282841904--8.984326040433697*\x)/1.259484958939297});
        \draw [line width=2.pt,domain=-13.0:-5.192322082401852] plot(\x,{(--35.52376725054282--7.709098243849907*\x)/-0.9492743759679314});
        \begin{Large}
        \draw (-13.253393330406864,5.547998749394442) node[anchor=north west] {Area($S(O,1.43,\pi/12) \cap S'$) = 8.91};
        \draw [fill=black] (-5.192322082401852,4.745049387656246) circle (2.5pt);
        \draw[color=black] (-5.500145489327644,4.27635330672618) node {$O$};
        \draw[color=qqwuqq] (-3.2870257971905843,5.4081481968675295) node {$\theta = 1.43$ rad};
        \draw [fill=uuuuuu] (-3.9328371234625554,13.729375428089941) circle (2.0pt);
        \draw[color=uuuuuu] (-3.0711283521199166,13.541011591120377) node {$K_2$};
        \draw [fill=uuuuuu] (-6.141596458369784,12.454147631506153) circle (2.0pt);
        \draw[color=uuuuuu] (-6.579589938740922,12.725114146049702) node {$K_3$};
        \draw [fill=uuuuuu] (-5.346852796853905,6.) circle (2.0pt);
        \draw[color=uuuuuu] (-5.9076241237173665,6.359558461304369) node {$K_4$};
        \draw [fill=uuuuuu] (-5.016394426465813,6.) circle (2.0pt);
        \draw[color=uuuuuu] (-4.339376237918641,6.554900355788326) node {$K_1$};
        \draw [fill=uuuuuu] (-17.320508075688775,6.) circle (2.0pt);
        \draw[color=uuuuuu] (-1.324590072899436,6.61938753642216) node {$\varepsilon_1$};
        \draw[color=uuuuuu] (-11.024590072899436,8.61938753642216) node {$\varepsilon_2$};
        \end{Large}
        \end{tikzpicture}
        } 
         &
        \scalebox{0.4}{
        \definecolor{ududff}{rgb}{0.30196078431372547,0.30196078431372547,1.}
        \definecolor{uuuuuu}{rgb}{0.26666666666666666,0.26666666666666666,0.26666666666666666}
        \definecolor{qqwuqq}{rgb}{0.,0.39215686274509803,0.}
        \definecolor{cqcqcq}{rgb}{0.7529411764705882,0.7529411764705882,0.7529411764705882}
        \begin{tikzpicture}[line cap=round,line join=round,>=triangle 45,x=1.0cm,y=1.0cm]
        \draw [color=cqcqcq,, xstep=2.0cm,ystep=2.0cm] (-13.,4.) grid (-0.5,14.);
        \clip(-13.,4.) rectangle (-0.5,14.);
        \draw [shift={(-5.192322082401852,4.745049387656246)},line width=2.pt,color=qqwuqq,fill=qqwuqq,fill opacity=0.10000000149011612] (0,0) -- (0.027829690692931185:0.5373931719486121) arc (0.027829690692931185:108.75423833971107:0.5373931719486121) -- cycle;
        \fill[line width=2.pt,color=ududff,fill=ududff,fill opacity=0.10000000149011612] (-5.618423944472327,6.) -- (-7.536415619792454,11.648848413187926) -- (-9.174140457368836,10.703307537354668) -- (-6.030987558512409,6.) -- cycle;
        \draw [line width=2.pt,domain=-13.:-0.5] plot(\x,{(--16.--0.5773502691896257*\x)/1.});
        \draw [line width=2.pt,domain=-13.:-0.5] plot(\x,{(--6.-0.*\x)/1.});
        \draw [line width=2.pt,color=ududff] (-5.618423944472327,6.)-- (-7.536415619792454,11.648848413187926);
        \draw [line width=2.pt,color=ududff] (-7.536415619792454,11.648848413187926)-- (-9.174140457368836,10.703307537354668);
        \draw [line width=2.pt,color=ududff] (-9.174140457368836,10.703307537354668)-- (-6.030987558512409,6.);
        \draw [line width=2.pt,color=ududff] (-6.030987558512409,6.)-- (-5.618423944472327,6.);
        \draw [line width=2.pt,domain=-13.0:-5.192322082401852] plot(\x,{(--24.72390852852829--6.90379902553168*\x)/-2.3440935373906013});
        \draw [line width=2.pt,domain=-13.0:-5.192322082401852] plot(\x,{(--12.04327052143445--5.958258149698422*\x)/-3.9818183749669833});
        \begin{Large}
        \draw (-13.253393330406864,5.547998749394442) node[anchor=north west] {Area($S(O,1.90,\pi/12) \cap S'$) = 6.5};
        \draw [fill=black] (-5.192322082401852,4.745049387656246) circle (2.5pt);
        \draw[color=black] (-5.500145489327644,4.27635330672618) node {$O$};
        \draw[color=qqwuqq] (-3.2870257971905843,5.4081481968675295) node {$\theta = 1.90$ rad};
        \draw [fill=uuuuuu] (-7.536415619792454,11.648848413187926) circle (2.0pt);
        \draw[color=uuuuuu] (-6.682410467614591,11.455926083959733) node {$K_2$};
        \draw [fill=uuuuuu] (-9.174140457368836,10.703307537354668) circle (2.0pt);
        \draw[color=uuuuuu] (-9.010487428531093,11.56246454205823)  node {$K_3$};
        \draw [fill=uuuuuu] (-6.030987558512409,6.) circle (2.0pt);
        \draw[color=uuuuuu] (-7.09548738381159,6.459558461304369) node {$K_4$};
        \draw [fill=uuuuuu] (-5.618423944472327,6.) circle (2.0pt);
        \draw[color=uuuuuu] (-4.9627523173790316,6.554900355788326) node {$K_1$};
        \draw [fill=uuuuuu] (-17.320508075688775,6.) circle (2.0pt);
        \draw[color=uuuuuu] (-1.324590072899436,6.61938753642216) node {$\varepsilon_1$};
        \draw[color=uuuuuu] (-5.024590072899436,12.61938753642216) node {$\varepsilon_2$};
        \end{Large}
        \end{tikzpicture}
         }
         &
         \scalebox{0.4}{
       \definecolor{ududff}{rgb}{0.30196078431372547,0.30196078431372547,1.}
        \definecolor{uuuuuu}{rgb}{0.26666666666666666,0.26666666666666666,0.26666666666666666}
        \definecolor{qqwuqq}{rgb}{0.,0.39215686274509803,0.}
        \definecolor{cqcqcq}{rgb}{0.7529411764705882,0.7529411764705882,0.7529411764705882}
        \begin{tikzpicture}[line cap=round,line join=round,>=triangle 45,x=1.0cm,y=1.0cm]
        \draw [color=cqcqcq,, xstep=2.0cm,ystep=2.0cm] (-13.62,4.) grid (-0.5,14.);
        \clip(-13.62,4.) rectangle (-0.5,14.);
        \draw [shift={(-5.192322082401852,4.745049387656246)},line width=2.pt,color=qqwuqq,fill=qqwuqq,fill opacity=0.10000000149011612] (0,0) -- (0.027829690692931185:0.5373931719486121) arc (0.027829690692931185:142.4104231498335:0.5373931719486121) -- cycle;
        \fill[line width=2.pt,color=ududff,fill=ududff,fill opacity=0.10000000149011612] (-6.822522059510523,6.) -- (-11.321613679445868,9.463463294510978) -- (-13.50437134818843,8.203247566886741) -- (-8.208693838820807,6.) -- cycle;
        \draw [line width=2.pt,domain=-13.62:-0.5] plot(\x,{(--16.--0.5773502691896257*\x)/1.});
        \draw [line width=2.pt,domain=-13.62:-0.5] plot(\x,{(--6.-0.*\x)/1.});
        \draw [line width=2.pt,color=ududff] (-6.822522059510523,6.)-- (-11.321613679445868,9.463463294510978);
        \draw [line width=2.pt,color=ududff] (-11.321613679445868,9.463463294510978)-- (-13.50437134818843,8.203247566886741);
        \draw [line width=2.pt,color=ududff] (-13.50437134818843,8.203247566886741)-- (-8.208693838820807,6.);
        \draw [line width=2.pt,color=ududff] (-8.208693838820807,6.)-- (-6.822522059510523,6.);
        \draw [line width=2.pt,domain=-13.62:-5.192322082401852] plot(\x,{(-4.584266616846456--4.718413906854733*\x)/-6.129291597044015});
        \draw [line width=2.pt,domain=-13.62:-5.192322082401852] plot(\x,{(-21.48500550744876--3.4581981792304957*\x)/-8.312049265786577});
        \begin{Large}
        \draw (-8.394590072899436,9.207998749394442) node[anchor=north west] {Area($S(O,2.48,\pi/12) \cap S'$) = 8.14};
        \draw [fill=black] (-5.192322082401852,4.745049387656246) circle (2.5pt);
        \draw[color=black] (-5.500145489327644,4.27635330672618) node {$O$};
        \draw[color=qqwuqq] (-3.2870257971905843,5.4081481968675295) node {$\theta = 2.48$ rad};
        \draw [fill=uuuuuu] (-11.321613679445868,9.463463294510978) circle (2.0pt);
        \draw[color=uuuuuu] (-10.293692583109262,9.478319211188817) node {$K_2$};
        \draw [fill=uuuuuu] (-13.50437134818843,8.203247566886741) circle (2.0pt);
        \draw[color=uuuuuu] (-13.131128530997934,8.768960224216638) node {$K_3$};
        \draw [fill=uuuuuu] (-8.208693838820807,6.) circle (2.0pt);
        \draw[color=uuuuuu] (-8.466555798483983,5.359558461304371) node {$K_4$};
        \draw [fill=uuuuuu] (-6.822522059510523,6.) circle (2.0pt);
        \draw[color=uuuuuu] (-6.166513022543922,6.554900355788328) node {$K_1$};
        \draw [fill=uuuuuu] (-17.320508075688775,6.) circle (2.0pt);
        \draw[color=uuuuuu] (-1.324590072899436,6.61938753642216) node {$\varepsilon_1$};
        \draw[color=uuuuuu] (-5.024590072899436,12.61938753642216) node {$\varepsilon_2$};
        \end{Large}
        \end{tikzpicture}
         }
    \end{tabular}
    \caption{An example that the area of the intersection $A(S(O,\theta,\pi/12)\cap S')$ is not a monotonic function. One might assume intuitively that as the sector $A(S(O,\theta,\pi/12))$ rotates counterclockwise then the area of intersection should decrease but there are cases where the function has local extreme points.}
    \label{fig:Local_extreme_points_problem}
\end{figure}
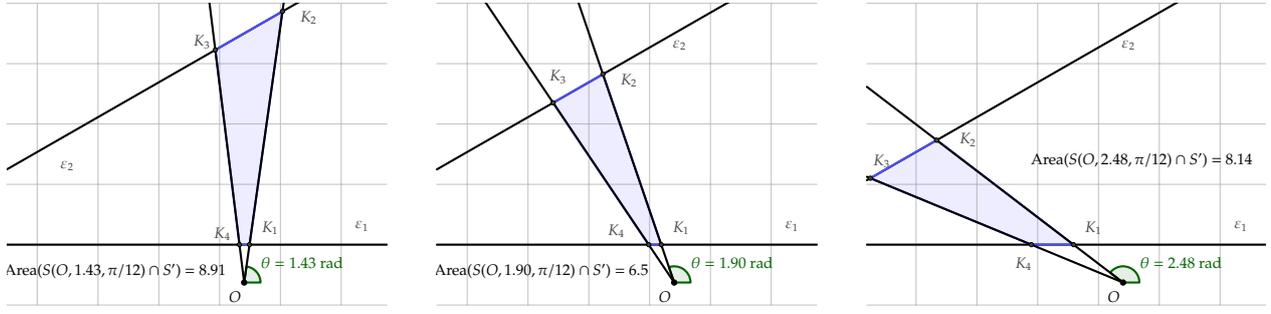

\subsection{The Intersection Area Function
}
The intersection of a fixed sector and a rotating one, when the rotating sector fully intersects the other one, is given in the following theorem.

\begin{theorem}
\label{theor:Analytical_Area}
Let two sectors on the plane, $S(C,\theta,\phi)$, $S(K,\theta_K,\phi_K)$ with $C \notin S(K,\theta_K,\phi_K)$, and $\mathcal{R}\subseteq[0,2\pi]\times(0,\pi)$. The area of the bounded intersection $S(C,\theta,\phi) \cap S(K,\theta_K,\phi_K)$ is 
\begin{align}
\label{eq:General_Case_Alt}
\hspace{-0.65cm}    &A(\theta,\phi)=\frac{d_1 \sin{\phi} ~\cos^2{(\theta_K+\phi_K)}}{2\sin{(\theta+\phi-\theta_K-\phi_K)}\sin{(\theta-\theta_K -\phi_K)}}-\frac{d_2 \sin{\phi} ~\cos^2{(\theta_K)} }{2\sin{(\theta+\phi-\theta_K)}\sin{(\theta-\theta_K)}} 
\hspace{-0.25cm} 
\end{align}
for every $(\theta, \phi) \in \mathcal{R}$,
where $\phi_K,\theta_K \in (-\pi/2,\pi/2)$, and $d_1$,$d_2 \in \mathbb{R}$, are constants representing distances (as in Figure~\ref{fig:General_Case2}(a)).
\end{theorem}

The proof of Theorem~\ref{theor:Analytical_Area} can be found in Section~\ref{sec:proof_calc}.
We derive this equation by expressing two of the four points of the intersection which is a quadrilateral, as the intersection of the left semi-line of the rotating sector with the left and right semi-lines of the static one. We do the same for the other two points of the quadrilateral, by using the right semi-line of the rotating sector. Then we use the shoelace formula to calculate the quadrilateral's area using the four points we identified and simplify the expression.
\footnote{Apart from the analytical proof of equation~(\ref{eq:General_Case_Alt}), various tests have been performed, in a simulation environment,  to affirm its validity.}
In the above theorem, if $\varepsilon_y$ is the vertical line that passes through $C=(x_C,y_C)$, and $E'$, $E$ are the intersections of $\varepsilon_y$ with the left and the right semi-line of $S(K,\theta_K,\phi_K)$ respectively,
then $d_1 = sign(x_C-x_K)|CE|^2$, $d_2 = sign(x_K-x_C)|CE'|^2$, where $sign(x)=1$ if $x\geq 0$, and $sign(x)=-1$, if $x<0$. 
Alternatively, a static sector consists of two intersecting lines. If a rotating sector intersects two parallel lines, then the analysis of Theorem~\ref{theor:Analytical_Area} is sound. 

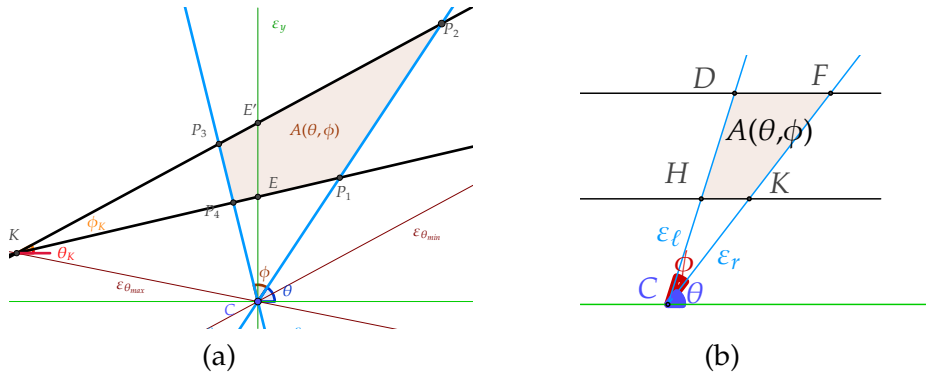
\begin{figure}[ht]
    \centering
    \begin{tabular}{c c}
         \scalebox{0.5}{
    \definecolor{dcrutc}{rgb}{0.8627450980392157,0.0784313725490196,0.23529411764705882}
\definecolor{ffxfqq}{rgb}{1.,0.4980392156862745,0.}
\definecolor{ffqqqq}{rgb}{1.,0.,0.}
\definecolor{xdxdff}{rgb}{0.49019607843137253,0.49019607843137253,1.}
\definecolor{yqqqqq}{rgb}{0.5019607843137255,0.,0.}
\definecolor{qqttcc}{rgb}{0.,0.2,0.8}
\definecolor{qqccqq}{rgb}{0.,0.8,0.}
\definecolor{zzttqq}{rgb}{0.6,0.2,0.}
\definecolor{qqzzff}{rgb}{0.,0.6,1.}
\definecolor{uuuuuu}{rgb}{0.26666666666666666,0.26666666666666666,0.26666666666666666}
\definecolor{qqzzqq}{rgb}{0.,0.6,0.}
\definecolor{ududff}{rgb}{0.30196078431372547,0.30196078431372547,1.}
\begin{tikzpicture}[line cap=round,line join=round,>=triangle 45,x=1.0cm,y=1.0cm]
\clip(-11.7,-0.5) rectangle (0.5,8.);
\draw [shift={(-5.159393740790174,0.21531696088107866)},line width=2.pt,color=zzttqq,fill=zzttqq,fill opacity=0.10000000149011612] (0,0) -- (56.71501006225292:0.44436403781217276) arc (56.71501006225292:103.78711748401521:0.44436403781217276) -- cycle;
\draw [shift={(-5.159393740790174,0.21531696088107866)},line width=2.pt,color=qqttcc,fill=qqttcc,fill opacity=0.10000000149011612] (0,0) -- (0.:0.44436403781217276) arc (0.:56.71501006225292:0.44436403781217276) -- cycle;
\fill[line width=2.pt,color=zzttqq,fill=zzttqq,fill opacity=0.10000000149011612] (-5.803401263375599,2.839793156357704) -- (-3.0039124901504244,3.4985976737639572) -- (-0.3207801583523553,7.585609162493171) -- (-6.184794752758505,4.394057817703224) -- cycle;
\draw [shift={(-11.506444202279491,1.4976943820935658)},line width=2.pt,color=ffqqqq,fill=ffqqqq,fill opacity=0.10000000149011612] (0,0) -- (0.:0.44436403781217276) arc (0.:13.242483943385635:0.44436403781217276) -- cycle;
\draw [shift={(-11.506444202279491,1.4976943820935658)},line width=2.pt,color=ffxfqq,fill=ffxfqq,fill opacity=0.10000000149011612] (0,0) -- (13.242483943385633:0.44436403781217276) arc (13.242483943385633:28.557705384660107:0.44436403781217276) -- cycle;
\draw [line width=2.pt,domain=-11.7:0.5] plot(\x,{(--76.79914092070155--5.386298304762439*\x)/9.8965451144985});
\draw [line width=2.pt,domain=-11.7:0.5] plot(\x,{(-27.584748764628948-1.5435769155688561*\x)/-6.559193404983679});
\draw [line width=0.4pt,color=qqzzqq] (-5.159393740790174,-0.5) -- (-5.159393740790174,8.);
\draw [line width=2.pt,color=qqzzff,domain=-11.7:0.5] plot(\x,{(--39.06807502424048--7.370292201612092*\x)/4.8386135824378185});
\draw [line width=2.pt,color=qqzzff,domain=-11.7:0.5] plot(\x,{(--56.19305667643323--11.004096104127795*\x)/-2.7002419311424806});
\draw [line width=0.4pt,color=qqccqq,domain=-11.7:0.5] plot(\x,{(--0.21531696088107866-0.*\x)/1.});
\draw [line width=0.4pt,color=yqqqqq,domain=-11.7:0.5] plot(\x,{(--5.249662424407623--1.282377421212487*\x)/-6.347050461489317});
\draw [line width=0.4pt,color=yqqqqq,domain=-11.7:0.5] plot(\x,{(--3.0233710280431114--0.5442604709467226*\x)/1.});
\draw [line width=2.pt,color=dcrutc] (-11.506444202279491,1.4976943820935658)-- (-10.627173350827759,1.4976943820935658);
\begin{large}
\draw[color=black] (-13.454650235760942,0.12609754063500522) node {$\varepsilon_1$};
\draw[color=black] (-12.921413390386334,1.8443051535087376) node {$\varepsilon_2$};
\draw [fill=ududff] (-5.159393740790174,0.21531696088107866) circle (2.5pt);
\draw[color=ududff] (-5.893055525657135,0.015006531181962171) node {$C$};
\draw [fill=uuuuuu] (-5.159393740790174,4.952143055386355) circle (2.0pt);
\draw[color=uuuuuu] (-5.359818680282528,5.302938581146811) node {$E'$};
\draw [fill=uuuuuu] (-5.159393740790174,2.991347636297303) circle (2.0pt);
\draw[color=uuuuuu] (-4.752521161939225,3.362548949366993) node {$E$};
\draw [fill=uuuuuu] (-0.3207801583523553,7.585609162493171) circle (2.5pt);
\draw[color=uuuuuu] (-0.06448056302080292,7.413667760754629) node {$P_2$};
\draw[color=qqzzff] (-6.359637765359917,-0.6145091890519485) node {$\varepsilon_r$};
\draw[color=qqzzff] (-4.063756903330358,-0.6145091890519485) node {$\varepsilon_l$};
\draw[color=zzttqq] (-4.974703180845312,0.9629831451812628) node {$\phi$};
\draw[color=qqttcc] (-4.382217797095747,0.5186191073690907) node {$\theta$};
\draw [fill=uuuuuu] (-6.184794752758505,4.394057817703224) circle (2.0pt);
\draw[color=uuuuuu] (-6.700316861015916,4.747483533881597) node {$P_3$};
\draw [fill=uuuuuu] (-5.803401263375599,2.839793156357704) circle (2.0pt);
\draw[color=uuuuuu] (-6.359637765359917,2.5997240177894305) node {$P_4$};
\draw [fill=uuuuuu] (-3.0039124901504244,3.4985976737639572) circle (2.0pt);
\draw[color=uuuuuu] (-2.8935982704249694,3.08852445938282) node {$P_1$};
\draw[color=zzttqq] (-3.6564232020025327,4.710453197397248) node {$A(\theta,\phi)$};
\draw [fill=uuuuuu] (-11.506444202279491,1.4976943820935658) circle (2.0pt);
\draw[color=uuuuuu] (-11.580915209652947,1.999832566742998) node {$K$};
\draw[color=yqqqqq] (-14.187850898151027,2.2960752586177793) node {$k$};
\draw[color=yqqqqq] (-6.929904947218872,-0.6515395255362961) node {$e111$};
\draw[color=ffqqqq] (-10.203314326466167,1.485165164087349) node {$\theta_K$};
\draw[color=ffxfqq] (-9.395872076184775,2.273893239711693) node {$\phi_K$};
\draw[color=qqzzqq] (-4.6,7.5) node {$\varepsilon_y$};
\draw[color=yqqqqq] (-8.5,0.6) node {$\varepsilon_{\theta_{max}}$};
\draw[color=yqqqqq] (-0.7,2) node {$\varepsilon_{\theta_{min}}$};
\end{large}    
\end{tikzpicture}
    }
         &
         \scalebox{1.4}{
\definecolor{zzttqq}{rgb}{0.6,0.2,0.}
\definecolor{qqccqq}{rgb}{0.,0.8,0.}
\definecolor{uuuuuu}{rgb}{0.26666666666666666,0.26666666666666666,0.26666666666666666}
\definecolor{qqzzff}{rgb}{0.,0.6,1.}
\definecolor{ccqqqq}{rgb}{0.8,0.,0.}
\definecolor{ududff}{rgb}{0.30196078431372547,0.30196078431372547,1.}
\begin{tikzpicture}[line cap=round,line join=round,>=triangle 45,x=1.0cm,y=1.0cm]
\clip(5.176585295153291,2.2719567019729383) rectangle (8.424729008837147,4.8577694349339495);
\draw [shift={(6.,2.5)},line width=2.pt,color=ccqqqq] (0,0) -- (52.5971881455121:0.27219081399585376) arc (52.5971881455121:72.59718814551206:0.27219081399585376) -- cycle;
\draw [shift={(6.,2.5)},line width=2.pt,color=ududff] (0,0) -- (0.:0.13609540699792688) arc (0.:52.59718814551214:0.13609540699792688) -- cycle;
\fill[line width=2.pt,color=zzttqq,fill=zzttqq,fill opacity=0.10000000149011612] (7.529270994752774,4.5) -- (6.6268698012980485,4.5) -- (6.313434900649024,3.5) -- (6.7646354973763865,3.5) -- cycle;
\draw [line width=0.4pt] (4.5,4.5)-- (8.,4.5);
\draw [line width=0.4pt] (4.5,3.5)-- (8.,3.5);
\draw [line width=0.4pt,color=qqzzff,domain=6.0:8.424729008837147] plot(\x,{(-2.828979958350093--0.6919509268573623*\x)/0.5290902411176317});
\draw [line width=0.4pt,color=qqzzff,domain=6.0:8.424729008837147] plot(\x,{(-4.335781599753609--0.831180700013098*\x)/0.26052104012999155});
\draw [shift={(6.,2.5)},line width=2.pt,color=ududff] (0.:0.13609540699792688) arc (0.:52.59718814551214:0.13609540699792688);
\draw [shift={(6.,2.5)},line width=2.pt,color=ududff] (0.:0.10660806881504273) arc (0.:52.59718814551214:0.10660806881504273);
\draw [line width=0.4pt,color=qqccqq,domain=5.176585295153291:8.424729008837147] plot(\x,{(--2.5-0.*\x)/1.});
\begin{scriptsize}
\draw (6.446218989196048,4.354609832111218) node[anchor=north west] {$A$($\theta$,$\phi$)};
\draw [fill=ududff] (6.,2.5) circle (0.5pt);
\draw[color=ududff] (5.811500492942095,2.6600253134519275) node {$C$};
\draw[color=ccqqqq] (6.156472225538364,2.8779746661834825) node {$\phi$};
\draw[color=qqzzff] (6.576099730448639,2.925214655563088) node {$\varepsilon_r$};
\draw[color=qqzzff] (6.013375346655687,3.1429673067598043) node {$\varepsilon_\ell$};
\draw [fill=uuuuuu] (6.313434900649024,3.5) circle (0.5pt);
\draw[color=uuuuuu] (6.08783441743484,3.7489852709703493) node {$H$};
\draw [fill=uuuuuu] (7.529270994752774,4.5) circle (0.5pt);
\draw[color=uuuuuu] (7.430839134682573,4.67018255621968) node {$F$};
\draw [fill=uuuuuu] (6.7646354973763865,3.5) circle (0.5pt);
\draw[color=uuuuuu] (7.050558801227791,3.6489852709703493) node {$K$};
\draw [fill=uuuuuu] (6.6268698012980485,4.5) circle (0.5pt);
\draw[color=uuuuuu] (6.3428072934929657,4.647018255621968) node {$D$};
\draw[color=ududff] (6.251739010436913,2.6099269626845083) node {$\theta$};
\end{scriptsize}
\end{tikzpicture}
}
         \\
         (a) & (b) 
    \end{tabular}
    \caption{ (a) The area $A(\theta,\phi)$ of intersection of two sectors $S(C,\theta,\phi)\cap S(K,\theta_K,\phi_K)$. The function $A(\theta,\phi) $ is defined between the lines $\varepsilon_{\theta_{max}}$, and $\varepsilon_{\theta_{min}}$.Finally, if $\varepsilon_y$ is a vertical line that passes through $C$, then $E' = \varepsilon_y \cap KP_2$, $E = \varepsilon_y \cap KP_1$ and $d_1=|CE'|^2$, $d_2=|CE|^2$.
    \\
    (b) The intersection of a rotating sector with two parallel lines, which can be considered a special case of Theorem~\ref{theor:Analytical_Area}, where equation~(\ref{eq:General_Case_Alt}) hols.}
    \label{fig:General_Case2}
\end{figure}

\begin{corollary}[Intersection of Two Parallel Lines and a Sector]
    If two parallel lines intersect with a sector, then equation~(\ref{eq:General_Case_Alt}) holds for $\phi_K=0$ and $\mathcal{R}=(\theta_k,\pi+\theta_k-\phi)\times (0,\pi)$ that is
    \begin{align}
        \label{eq:Parallel}
        \hspace{-0.65cm}    &A(\theta,\phi)=\frac{(d_1-d_2) \sin{\phi} ~\cos^2{(\theta_K)}}{2\sin{(\theta+\phi-\theta_K)}\sin{(\theta-\theta_K)}} 
        \hspace{-0.25cm} 
    \end{align}
\end{corollary}

In Proposition~\ref{proposition:difcil_Ath}, we show that the original function $A_\phi(\theta)$ could be standardised and expressed as an exponential polynomial function (polynomials with non-integer powers). However, the direct maximisation of these non-convex functions is difficult, see~\cite{lax1948quotient}. The main difficulty is that there are no constructive criteria to check the number of possible solutions of $dA_{\phi}/d\theta =0$, which means there is no guarantee of finding the global maximum value within a given precision and computation time by applying naively general numerical methods. 
The domain of function $A(\theta,\phi)$ is $\mathcal{R}= \mathcal{D} \times \mathcal{I} \subseteq[0,2\pi]\times(0,\pi)$.
We will denote the restriction of $A(\theta,\phi)$ at $\mathcal{D}$ and $\mathcal{I}$ respectively, as $A_\phi: \mathcal{D} \rightarrow \mathbb{R}$, and $A_\theta:\mathcal{I}\rightarrow \mathbb{R}$.

\begin{proposition}
\label{proposition:difcil_Ath}
    The function $A_\phi(\theta)$ is a rational function
    of the form $P(x)/Q(x)$, where $P(x)$, and $Q(x)$ are exponential polynomials (polynomials with non-integer powers).
\end{proposition}

The proof of Proposition~\ref{proposition:difcil_Ath} can be found in Section~\ref{sec:proof_calc}.
Even though it is hard to find the local extreme points of $A_\phi(\theta)$, 
we found a way to identify its global maximum indirectly following the analysis of the function $A_\theta(\phi)$. 
The function $A_\theta(\phi)$ has a symmetry that allows the cancellation of terms and
gives us the possibility
to calculate the explicit analytical form of its extreme points, see Lemma~\ref{lemma:Der_Phi_Roots}.

\begin{lemma}
\label{lemma:Der_Phi_Roots}
The function $A_\theta(\phi)$  has at most two local extreme points in $\mathcal{I} \subseteq (0,\pi)$, and they can be explicitly calculated.
\end{lemma}

\begin{proof}
\begin{align*}
    \frac{\partial A}{\partial\phi} =& 
    \frac{d_1^2\cos^2{\omega} \sin{(\theta-\omega)} (\cos{\phi}\sin{(\theta+\phi-\omega)} -\sin{\phi}\cos{(\theta+\phi-\omega)} )}
    {2\sin^2{(\theta+\phi-\omega)}\sin^2{(\theta-\omega)}}
    -
    \\
    &- 
    \frac{d_2^2\cos^2{\beta}\sin{(\theta-\beta)}(\cos{\phi}\sin{(\theta+\phi-\beta)}-\sin{\phi}\cos{(\theta+\phi-\beta)})}
    {2\sin^2{(\theta+\phi-\beta)}\sin^2{(\theta-\beta)}}
\end{align*}

\begin{align*}
    \frac{\partial A}{\partial \phi} =& 
    \frac{d_1^2\cos^2{\omega} \sin^2{(\theta-\omega)}}
    {2\sin^2{(\theta+\phi-\omega)}\sin^2{(\theta-\omega)}}
    -
    \frac{d_2^2\cos^2{\beta}\sin^2{(\theta-\beta)}}
    {2\sin^2{(\theta+\phi-\beta)}\sin^2{(\theta-\beta)}}
\end{align*}
\begin{align*}
    \frac{\partial A}{\partial \phi}=0 \Rightarrow
    \frac{d_1^2\cos^2{\omega}}
    {2\sin^2{(\theta+\phi-\omega)}}
    =
    \frac{d_2^2\cos^2{\beta}}
    {2\sin^2{(\theta+\phi-\beta)}}
\end{align*}
\begin{align*}
    \frac{\sin^2{(\theta+\phi-\beta)}}
    {\sin^2{(\theta+\phi-\omega)}}
    =
    \frac{d_2^2\cos^2{\beta}}
    {d_1^2\cos^2{\omega}}
    \Rightarrow
    \left| \frac{\sin{(\theta+\phi-\beta)}}
    {\sin{(\theta+\phi-\omega)}} \right|
    =
    \frac{d_2\cos{\beta}}
    {d_1\cos{\omega}}
\end{align*}
\begin{align*}
    \frac{\sin{\phi}\cos{(\theta-\beta)}+\cos{\phi}\sin{(\theta-\beta)}}
    {\sin{\phi}\cos{(\theta-\omega)}+\cos{\phi}\sin{(\theta-\omega)}} 
    =
    \pm
    \frac{d_2\cos{\beta}}
    {d_1\cos{\omega}}
\end{align*}
If $\phi \neq \pi/2$ we have the following two solution
\begin{align*}
    \phi_1
    =
    \arctan{ \left(
    \frac{d_2\cos{\beta}\sin{(\theta-\omega)} - d_1\cos{\omega}\sin{(\theta-\beta)}}
    {d_1\cos{\omega}\cos{(\theta-\beta)}-d_2\cos{\beta}\cos{(\theta-\omega)}}
    \right)}
    \\
    \phi_2
    = \arctan{ \left(
    - \frac{d_2\cos{\beta}\sin{(\theta-\omega)} + d_1\cos{\omega}\sin{(\theta-\beta)}}
    {d_1\cos{\omega}\cos{(\theta-\beta)}+d_2\cos{\beta}\cos{(\theta-\omega)}}
    \right)}
\end{align*}
Note that the tangent function defined in $[0,\pi/2)\cup(\pi/2,\pi]$ is an injection which guarantees that $\phi_1$ and $\phi_2$ are unique.
Moreover, if $\phi = \pi/2$ then
\begin{align*}
    \frac{\partial A}{\partial \phi} = 0 \Leftrightarrow
    \frac{\cos{(\theta-\beta)}}{\cos{(\theta-\omega)}} = \pm 
    \frac{d_2\cos{\beta}}
    {d_1\cos{\omega}}
\end{align*}
\end{proof}

From Lemma~\ref{lemma:Der_Phi_Roots} we calculate the roots of the equation $dA_\theta/d\phi=0$, $\phi_1<\phi_2$. Now in the interval $(\phi_1,\phi_2)$ function $dA_\theta/d\phi$ is either positive or negative, which means that $A_\theta(\phi)$ is strictly increasing or decreasing. In the following section, we express the area of intersection as a linear combination of $A_\theta(\phi)$ functions. Using the local extreme points of each $A_\theta(\phi)$, we can identify the intervals where there is at most one solution of the linear combination, which leads to an effective method of finding a global maximum for the original $A_\phi(\theta)$ function (see Lemma~\ref{lemma:A_decomp_AL_AR}).

\subsection{Approximating the Maximum Area Under Restricted Rotations}


The objective of maximising the area of the intersection of two sectors without restricting the domain of rotations is an ill-posed optimisation problem because there are unbounded intersections which means that the maximum is infinity. Even if we disregard those, the natural domain of rotations where equation~(\ref{eq:General_Case_Alt}) is well-defined is an open set $(\theta_{min},\theta_{max})$, which means that we can create a strictly increasing sequence of $A_{\phi}(\theta_i)$ which tends to infinity as $\theta_i$ tends to either $\theta_{min}$ or $\theta_{max}$. For these reasons, we study the maximisation of $A(\theta,\phi)$ under restricted rotations, i.e. we consider that $\theta$ belongs in a closed and bounded subset of $\mathbb{R}$, which guarantees the existence of a maximum. 

\begin{problem}
Given an interval $[a,b]$, a fixed sector $S(K,\theta_K,\phi_K)$ and a sector $S(C,\theta,\phi)$ with the centre $C \notin S(K,\theta_K,\phi_K)$, calculate the area of the intersection when $\theta \in [a,b]$, and $S(C,\theta,\phi)$ fully intersects $S(K,\theta_K,\phi_K)$.
\end{problem}

From equation~(\ref{eq:General_Case_Alt}), one can verify that $A_\phi(\theta)$ is not a convex function which means that this function may have multiple extreme points inside a given interval $[a,b]$.
Since finding the local extreme points of $A_\phi(\theta)$ analytically is a non-trivial problem (see Proposition~\ref{proposition:difcil_Ath}),  we first conceptualise the change in direction as an increase or a decrease of two different sectors to express $A_{\phi}(\theta_i)$ as a combination of $A_{\theta}(\phi)$ functions. 
Secondly, we apply numerical methods to approximate the solutions of equations. The method taken into consideration is the Newton Raphson method~\cite{NR1,NR2,cont_opt}.
It is easy to modify Newton Raphson to return a value equivalent to a negative value if it does not converge after a constant number of iterations.
Keep in mind that if Newton Raphson converges, then the time complexity needed to approximate the solution of an equation $f(x)=0$ up to $\varepsilon>1$ accuracy is $F(\varepsilon)\cdot \log{\varepsilon}$, where $F(\varepsilon)$ is the complexity of computing $f'/f$, up to $\varepsilon$ precision, that is $|apx-opt|<10^{-\varepsilon}$. 

We initially focus on the case where the angle of the rotation $\theta$ is bounded by the inner angle $\phi$, i.e. $\theta \in [\theta_0-\phi,\theta_0+\phi]$ because it allows us to express $A_\phi(\theta)$ with only two functions $A_\theta$ and a constant $A(\theta_0,\phi)$, see Lemma~\ref{lemma:A_decomp_AL_AR} and Figure~\ref{fig:Restricted_Rotation}. In the following lemma not only do we express the function $A_\phi$ as a summation of $A_\theta$ but we also divide the domain $[\theta_0,\theta_0+\phi]$ into a finite number of intervals where in each there exists at most one point that is a root of the first derivative of the function of the area to identify every possible local maximum. In the end, we obtain the maximum by selecting the maximum out of all local maximums. 

\begin{figure}[ht]
    \centering
    \scalebox{0.9}{
\definecolor{qqzzcc}{rgb}{0.,0.6,0.8}
\definecolor{zzttqq}{rgb}{0.6,0.2,0.}
\definecolor{ttzzqq}{rgb}{0.2,0.6,0.}
\definecolor{xfqqff}{rgb}{0.4980392156862745,0.,1.}
\definecolor{qqzzqq}{rgb}{0.,0.6,0.}
\definecolor{qqwuqq}{rgb}{0.,0.39215686274509803,0.}
\definecolor{uuuuuu}{rgb}{0.26666666666666666,0.26666666666666666,0.26666666666666666}
\definecolor{yqqqqq}{rgb}{0.5019607843137255,0.,0.}
\definecolor{qqzzff}{rgb}{0.,0.6,1.}
\definecolor{ududff}{rgb}{0.30196078431372547,0.30196078431372547,1.}
\begin{tikzpicture}[line cap=round,line join=round,>=triangle 45,x=1.0cm,y=1.0cm]
\clip(12.3,21.2) rectangle (17.5,27.);
\draw [shift={(15.39143,21.77)},line width=0.pt,color=qqwuqq,fill=qqwuqq,fill opacity=0.10000000149011612] (0,0) -- (0.:0.3783994067571905) arc (0.:71.06133972932105:0.3783994067571905) -- cycle;
\fill[line width=2.pt,color=xfqqff,fill=xfqqff,fill opacity=0.15000000596046448] (17.067343809321034,26.654187191531765) -- (15.948103878432178,23.39233845904057) -- (15.591644605021639,23.38396346006027) -- (15.845512138198584,25.430432158979073) -- (15.8903932189,25.79222606653287) -- cycle;
\fill[line width=2.pt,color=ttzzqq,fill=ttzzqq,fill opacity=0.10000000149011612] (15.324590990652041,25.37785053610977) -- (14.67280937156663,24.90050643839409) -- (15.023999567315155,23.370626658449012) -- (15.36162984418431,23.378559270770875) -- cycle;
\draw [shift={(15.39143,21.77)},line width=0.pt,color=zzttqq] (0,0) -- (0.:0.9459985168929764) arc (0.:82.92849826272506:0.9459985168929764) -- cycle;
\draw [shift={(15.39143,21.77)},line width=0.8pt,color=qqzzcc] (0,0) -- (71.0613397293205:0.693732245721516) arc (71.0613397293205:91.06133972932076:0.693732245721516) -- cycle;
\draw [line width=0.4pt,domain=12.502984045176737:17.5] plot(\x,{(--175.40641212945064--9.075635826622793*\x)/12.39217713513895});
\draw [line width=0.4pt,domain=12.502984045176737:17.5] plot(\x,{(--302.53507800813094--0.30880890928690263*\x)/13.143619442339164});
\draw [line width=0.4pt,color=qqzzff,domain=15.39143:17.5] plot(\x,{(-16.30669001917653--2.0585413370403103*\x)/0.7063484093700438});
\draw [line width=0.4pt,color=qqzzff,domain=12.3:15.39143] plot(\x,{(-34.36906368084209--2.1759814882101374*\x)/-0.04031221514744665});
\draw [line width=0.4pt,dash pattern=on 1pt off 1pt,color=yqqqqq,domain=15.39143:17.5] plot(\x,{(-46.45391719609205--3.6604321589790736*\x)/0.4540821381985847});
\draw [line width=0.4pt,dash pattern=on 1pt off 1pt,color=yqqqqq,domain=12.3:15.39143] plot(\x,{(-73.29754003887274--3.594986326668373*\x)/-0.8252438971520135});
\draw [line width=0.4pt,color=qqzzqq,domain=12.3:17.5] plot(\x,{(--35.971659500000044-0.*\x)/1.65235});
\draw [shift={(15.39143,21.77)},->,line width=0.pt,color=zzttqq] (0.:0.9459985168929764) arc (0.:82.92849826272506:0.9459985168929764);
\draw [shift={(15.39143,21.77)},line width=0.8pt,color=qqzzcc] (71.0613397293205:0.693732245721516) arc (71.0613397293205:91.06133972932076:0.693732245721516);
\draw [shift={(15.39143,21.77)},line width=0.8pt,color=qqzzcc] (71.0613397293205:0.6306656779286509) arc (71.0613397293205:91.06133972932076:0.6306656779286509);
\begin{scriptsize}
\draw [fill=black] (12.502984045176737,23.311395481138817) circle (2.5pt);
\draw[color=black] (12.603890553645327,23.797008053143877) node {$K$};
\draw[color=black] (19.19434688799973,27.864801675783657) node {$\varepsilon_{1}$};
\draw[color=black] (18.626747777863944,23.387075362490258) node {$\varepsilon_2$};
\draw [fill=ududff] (15.39143,21.77) circle (2.5pt);
\draw[color=ududff] (14.949966875539907,21.577064866835038) node {$C$};
\draw[color=qqzzff] (17.247364057636563,25.876275065868852) node {$\varepsilon_{r_1}$};
\draw[color=qqzzff] (14.998646015337833,26.79070139636459) node {$\varepsilon_{\ell_1}$};
\draw[color=yqqqqq] (16.39237826105935,26.722273582761823) node {$\varepsilon_{r_2}$};
\draw[color=yqqqqq] (13.955820278084901,26.707847252266082) node {$\varepsilon_{\ell_2}$};
\draw [fill=uuuuuu] (15.8903932189,25.79222606653287) circle (2.0pt);
\draw[color=uuuuuu] (15.618472494144276,26.02956455301129) node {$B$};
\draw [fill=uuuuuu] (17.067343809321034,26.654187191531765) circle (2.0pt);
\draw[color=uuuuuu] (17.341057400857276,26.56113673940852) node {$A$};
\draw [fill=uuuuuu] (15.324590990652041,25.37785053610977) circle (2.0pt);
\draw[color=uuuuuu] (15.05087338400849,25.739458341164113) node {$D$};
\draw [fill=uuuuuu] (14.67280937156663,24.90050643839409) circle (2.0pt);
\draw[color=uuuuuu] (14.357141138286975,25.197085858145474) node {$E$};
\draw[color=qqwuqq] (16.003178557680755,21.684278032082908) node {$\theta_0$};
\draw[color=qqzzqq] (19.54121301086049,21.98069090070937) node {$f$};
\draw [fill=uuuuuu] (15.948103878432178,23.39233845904057) circle (2.0pt);
\draw[color=uuuuuu] (16.17345829072149,23.191569002332376) node {$F$};
\draw [fill=uuuuuu] (15.591644605021639,23.38396346006027) circle (2.0pt);
\draw[color=uuuuuu] (15.63108580770285,23.191569002332376) node {$G$};
\draw [fill=uuuuuu] (15.36162984418431,23.378559270770875) circle (2.0pt);
\draw[color=uuuuuu] (15.265299714504232,23.141115748098084) node {$H$};
\draw [fill=uuuuuu] (15.023999567315155,23.370626658449012) circle (2.0pt);
\draw[color=uuuuuu] (14.798607112837031,23.191569002332376) node {$I$};
\draw[color=xfqqff] (16.17345829072149,25.247539112379766) node {$R$};
\draw[color=ttzzqq] (14.96258018909848,24.780846510712568) node {$L$};
\draw[color=zzttqq] (16.438337875451523,22.321250366790842) node {$\theta$};
\draw[color=qqzzcc] (15.784258587342895,22.321250366790842) node {$\phi$};
\end{scriptsize}
\end{tikzpicture}
    }
    \caption{The sector $S(C,\theta_0,\phi)$ has as borders the blue lines $\varepsilon_{\ell_1}$, and $\varepsilon_{r_1}$ while the sector $S(C,\theta,\phi)$ has as borders the red ones $\varepsilon_{\ell_2}$, and $\varepsilon_{r_2}$. The line $\varepsilon_{\ell_1}$ partitions $S(C,\theta,\phi)\cap S(K,\theta_K,\phi_K)$ into two quadrilaterals.}
    \label{fig:Restricted_Rotation}
\end{figure}
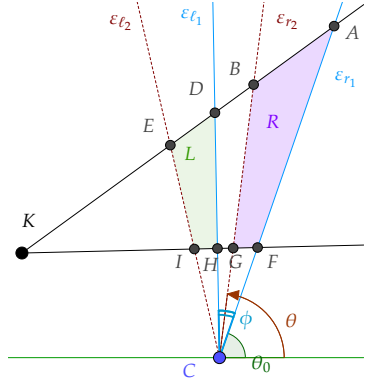

\begin{lemma}
\label{lemma:A_decomp_AL_AR}
    For $\phi \in (0,\pi)$, and $\theta \in [\theta_0,\theta_0+\phi]$, the function $A_\phi(\theta)$ is expressed as:
    \begin{align}
    \label{eq:A_L_N_A_R}
        A_\phi(\theta) = A(\theta_0,\phi) + A_\theta^L(\theta-\theta_0) -  A_\theta^R(\theta-\theta_0),
    \end{align}
    where $A(\theta_0,\phi)$ is constant. The maximum value of $A_\phi(\theta)$ can be approximated with precision $\varepsilon>1$, in time $\mathcal{O}(\log{\varepsilon})$.
\end{lemma}

\begin{proof}
Let $S'$ be a fixed sector, a rotating one at direction $\theta_0$, $S(C,\theta_0,\phi)$ with right and left semi-lines $\varepsilon_{r_1}$, $\varepsilon_{\ell_1}$  and its rotation at $\theta \in [\theta_0,\theta_0+\phi]$, $S(C,\theta,\phi)$ with right and left semi-lines $\varepsilon_{r_2}$, $\varepsilon_{\ell_2}$, respectively.
We can express the rotated intersection at direction $\theta$ by using the initial one, see Figure~\ref{fig:Restricted_Rotation},
\begin{align*}
    S[C,\varepsilon_{r_2},\varepsilon_{\ell_2}]\cap S' = \left(~ (S[C,\varepsilon_{r_1},\varepsilon_{\ell_1}]\cap S') ~ \cup ~S[C,\varepsilon_{\ell_2},\varepsilon_{\ell_1}]~ \right) \setminus S[C,\varepsilon_{r_2},\varepsilon_{r_1}] \\
    S(C,\theta,\phi)\cap S' = \left(~ (S(C,\theta_0,\phi)\cap S') ~ \cup ~S(C,\theta_0+\phi,\theta-\theta_0)~ \right) \setminus S(C,\theta_0,\theta-\theta_0)
\end{align*}
This means that the area of the intersection is expressed 
\begin{align*}
    A_\phi(\theta) = A(\theta_0,\phi) + A_{({\theta_0}+\phi)}(\theta-\theta_0) -  A_{\theta_0}(\theta-\theta_0), 
    &&
    \theta \in [\theta_0,\theta_0+\phi]
\end{align*}

\begin{table}[ht]
    \centering
    \begin{tabular}{|c| c | c l | c l | c l | c l | c l | c |}
        \hline
             $\phi$ & \multicolumn{2}{ c }{a} & \multicolumn{2}{r}{$\phi_1$} & \multicolumn{2}{r}{$\phi_2$} & \multicolumn{2}{r}{$\phi_3$} & \multicolumn{2}{r}{$\phi_4$} & \multicolumn{2}{c|}{b}
             \\
             \hline
              $dA^L/d\phi$ &   &   & $+$ &   & $-$ &   & $+$ &   & $+$ &   & $+$ &   \\
              $A_\phi^L$ &   &   & $\nearrow$ &   & $\searrow$ &   & $\nearrow$ &   & $\nearrow$ &   & $\nearrow$ &  \\
              $dA^R/d\phi$ &   &   & $-$ &   & $-$ &   & $-$ &   & $+$ &   & $-$ &   \\
              $A_\phi^R$ &   &   & $\searrow$ &   & $\searrow$ &   & $\searrow$ &   & $\nearrow$ &   &  $\searrow$ &
              \\
            \hline
    \end{tabular}
    \caption{Assume that $\phi_1$, and $\phi_2$ are the roots of $dA^L/d\phi=0$ and $\phi_3$, and $\phi_4$ are the roots of $dA^R/d\phi=0$. The sign of the respective derivative changes only if $\phi$ crosses one of its roots. Worst-case scenario we need to search for a local extreme point in every interval. }
    \label{tab:root_signs}
\end{table}

To find the local extreme points of the function $A_\phi(\theta)$ we use the above equation and the fact that the function $A_\theta(\phi)$ has two local extreme points (Lemma~\ref{lemma:Der_Phi_Roots}). 

Let $f$ be a continuous function in an interval $[a,b] \subseteq \mathcal{R}$, and $x_0 \in [a,b]$ be the only root of $f$ in $[a,b]$. From the intermediate value Theorem~\cite{spivak2019calculus} it follows that the sign of $f$ does not change sign inside intervals $[a.x_0]$ and $[x_0,b]$. 
Let $\theta_{1},\theta_{2}$ be the roots of $d A^L_\theta/d \phi=0$, and $\theta_{3},\theta_{4}$ be the roots of $d A^R_\theta/d \phi=0$. 
As mentioned above from the intermediate value Theorem~\cite{spivak2019calculus},
the values $\theta_1,\ldots,\theta_4$ partition the domain $[\theta_0,\theta_0+\phi]$ into at most five intervals where the functions $f_L = d A^L_\theta/d \phi$, and $f_R=d A^R_\theta/d \phi$ will be either positive or negative. An example is shown in Table~\ref{tab:root_signs} of how the monotonicity of $A^L_\phi$, and $A^R_\phi$ should remain intact inside the intervals $[a,\phi_1],[\phi_1,\phi_2],[\phi_2,b]$, and $[\phi_2,\phi_3],[\phi_3,\phi_4],[\phi_4,b]$ respectively.

Without loss of generality let's assume that $\theta_1<\ldots<\theta_4$, which partitions the domain $[\theta_0,\theta_0+\phi]$ in at most five intervals $[\theta_0,\theta_1]\cup [\theta_1,\theta_2] \cup [\theta_2,\theta_3] \cup [\theta_3,\theta_4] \cup [\theta_4,\theta_0+\phi]$.
By running a modified Newton Raphson in the intervals where $f_L\cdot f_R<0$ which returns a negative number if it does not converge after a constant number of iterations,  we can find all the values $r_1,\ldots,r_k$, $k\leq5$ of possible local maximum points of equation~(\ref{eq:A_L_N_A_R}). 
The local maximum of a function inside a closed given interval is either at a root of the derivative of said function or it is at the boundaries of the interval. 
Also by checking the edges of the interval $\theta_0$,$\theta_0+\phi$ we obtain the maximum out of the set of values $\{A_\phi(r_1),\ldots,A_\phi(r_k),A_\phi(\theta_0),A_\phi(\theta_0+\phi) \}$
The running time is $\mathcal{O}(\log{\varepsilon})$ because we run at most five times the Newton Raphson method and to do so, we can evaluate the derivative of $f_L+f_R$ in constant time by plugging the analytical formula. Furthermore, all the rest of the evaluations to check can also be done in constant time.
\end{proof}
Next, we show how to find a maximal intersection for unrestricted rotation with a direction $\theta \in [a,b]\supseteq [\theta_0,\theta_0+\phi]$ 

\begin{theorem}
\label{Theor:Inter_Two_Sectors}
Given an interval $[a,b]\subseteq[0,\pi]$ with $z=|b-a|$, and two sectors $S'$ and $S(C,\theta,\phi)$;
the direction of the maximum area of intersection $S(C,\theta,\phi) \cap S'$ where $\theta \in [a,b]$ can be $\varepsilon$-approximated in time $\mathcal{O}((z~\log{\varepsilon})/\phi)$.
\end{theorem}
\begin{proof}
We can partition the interval $[a,b]$ into $k>1$ intervals of length $\phi$, that is $[a,b] = \left[a,a+\phi\right]\cup\ldots\cup \left[a+(k-1)\phi,b \right]$.
For each interval $[a+i\phi,a+(i+1)\phi]$, $i\in \{0,\ldots,k-1\}$ we can find all the local extreme points using Lemma~\ref{lemma:A_decomp_AL_AR}, we run Newton Raphson up to five times and then we select the maximum value $M_k$. Then we select the maximum $\max{(M_0,\ldots,M_{k-1})}$. 
If the length of the given interval $[a,b]$ is $z$  then we would have $\lceil z/\phi \rceil$ intervals, and in each interval, we run at most 5 times Newton Raphson with $\varepsilon$ accuracy. Given the evaluation of the function of the area and its derivative and constant time, this means that in a worst-case scenario, we would have $\mathcal{O}(z~\log{\varepsilon}/\phi)$.
\end{proof}


\subsection{ \texorpdfstring{$\mathcal{LMR}$}{LMR} Intersection and the Global Objective Function}
\label{sec:Equiv_Classes}

In this section, we decompose the area of intersection of a polygon $\mathcal{P}$ and a sector $S[C,\varepsilon_r,\varepsilon_\ell]$ as a summation of multiple areas of intersections of two sectors. 
Notice that if $\mathcal{P} \cap S[C,\varepsilon_r,\varepsilon_\ell]$ does not contain any vertices of $\mathcal{P}$, then this case is identical to the intersection of two sectors.

In the case that $\mathcal{P} \cap S[C,\varepsilon_r,\varepsilon_\ell]$ contains one vertex $P_0 \in \mathcal{P}$ or it contains two vertices colinear with $C$, then we can express the area as $Area(\mathcal{P} \cap S[C,\varepsilon_r,\varepsilon_\ell])$ as $Area(\mathcal{P} \cap S[C,\varepsilon_r,CP_0])+Area(\mathcal{P} \cap S[C,CP_0,\varepsilon_\ell])$, see Figure~\ref{fig:LMR_Cases}(a). We will refer to $Area(\mathcal{P} \cap S[C,\varepsilon_r,CP_0])$, $Area(\mathcal{P} \cap S[C,CP_0,\varepsilon_\ell])$ as right and left respectively. Now we can use equation~(\ref{eq:General_Case_Alt}), so $Area(\mathcal{P} \cap S[C,\varepsilon_r,\varepsilon_\ell])=A^L(\theta_1,\phi_1)+A^R(\theta_2,\phi_2)$.

If $\mathcal{P} \cap S[C,\varepsilon_r,\varepsilon_\ell]$ contains two non colinear vertices $P_1,P_2 \in \mathcal{P}$, then using the lines $CP_1$ and $CP_2$ we can express the area of intersection of $Area(\mathcal{P} \cap S[C,\varepsilon_r,\varepsilon_\ell])$ as  $Area(\mathcal{P} \cap S[C,\varepsilon_r,CP_2])+Area(\mathcal{P} \cap S[C,CP_1,\varepsilon_\ell]) + Area(\mathcal{P} \cap S[C,CP_1,CP_2] )$ (see Figure~\ref{fig:LMR_Cases}(b)). The area of $S[C, CP_1, CP_2]$ remains constant for certain rotations and will call it  the middle area. Similarly using equation~(\ref{eq:General_Case_Alt}), the area of intersection is $Area(\mathcal{P} \cap S[C,\varepsilon_r,\varepsilon_\ell])=A^L(\theta_1,\phi_1)+A^R(\theta_2,\phi_2)+A^M$ where $A^M$ is a constant unless the number of the vertices that $\mathcal{P} \cap S[C,\varepsilon_r,\varepsilon_\ell]$ contains, changes. The same argument can be said in the case that $\mathcal{P} \cap S[C,\varepsilon_r,\varepsilon_\ell]$ contains $P_1,\ldots,P_k$ vertices $k>2$. The only difference is that the middle area will be $\sum_{i=1}^{k-1}{Area(S[C,P_i,P_{i+1}])}$.
This is a different decomposition from the one we presented in the previous section, which enables the identification of the extreme points.

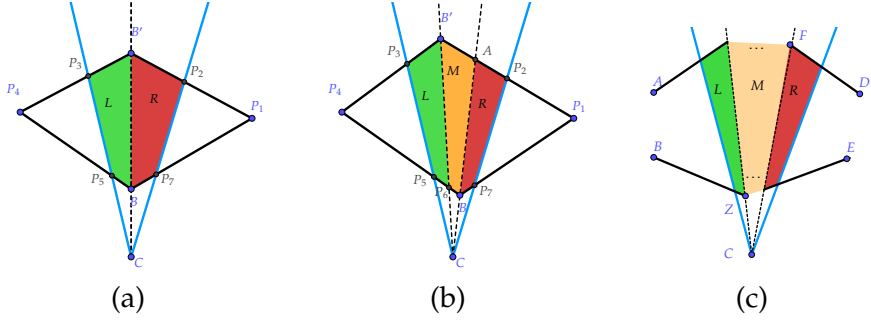
\begin{figure}[ht]
    \centering
    \begin{tabular}{c c c}
         \scalebox{0.45}{
         \definecolor{qqccqq}{rgb}{0.,0.8,0.}
\definecolor{ccqqqq}{rgb}{0.8,0.,0.}
\definecolor{uuuuuu}{rgb}{0.26666666666666666,0.26666666666666666,0.26666666666666666}
\definecolor{qqzzff}{rgb}{0.,0.6,1.}
\definecolor{ududff}{rgb}{0.30196078431372547,0.30196078431372547,1.}
\begin{tikzpicture}[line cap=round,line join=round,>=triangle 45,x=1.0cm,y=1.0cm]
\clip(2.3,1.59) rectangle (10.,9.5);
\fill[line width=2.pt,color=ccqqqq,fill=ccqqqq,fill opacity=0.75] (7.552008002205401,7.158268508597804) -- (6.,8.) -- (6.,4.) -- (6.731495463614648,4.431205257178962) -- cycle;
\fill[line width=2.pt,color=qqccqq,fill=qqccqq,fill opacity=0.699999988079071] (6.,8.) -- (4.74472582855669,7.329028416415277) -- (5.436147184634907,4.393730185893523) -- (6.,4.) -- cycle;
\draw [line width=2.pt,color=qqzzff,domain=6.0:10.0] plot(\x,{(-30.411428277928287--5.633577322436402*\x)/1.6950178283450619});
\draw [line width=2.pt,color=qqzzff,domain=2.3:6.0] plot(\x,{(-37.05570026152593--5.726329989586373*\x)/-1.348860162003846});
\draw [line width=1.2pt,dash pattern=on 3pt off 3pt] (6.,2.) -- (6.,9.5);
\draw [line width=1.2pt,dash pattern=on 3pt off 3pt] (6.,2.) -- (6.,9.5);
\draw [line width=2.pt] (2.755372967829234,6.265675668755097)-- (6.,8.);
\draw [line width=2.pt] (6.,8.)-- (9.534085607183261,6.083288781597591);
\draw [line width=2.pt] (9.534085607183261,6.083288781597591)-- (6.,4.);
\draw [line width=2.pt] (6.,4.)-- (2.755372967829234,6.265675668755097);
\begin{normalsize}
\draw [fill=ududff] (6.,2.) circle (2.5pt);
\draw[color=ududff] (6.226153669758307,1.8224351832384242) node {$C$};
\draw [fill=ududff] (6.,4.) circle (2.5pt);
\draw[color=ududff] (6.077754362065961,3.6362044994782092) node {$B$};
\draw [fill=ududff] (6.,8.) circle (2.5pt);
\draw[color=ududff] (6.176687233860858,8.632314525120526) node {$B'$};
\draw [fill=ududff] (2.755372967829234,6.265675668755097) circle (2.5pt);
\draw[color=ududff] (2.5573930073641966,6.925722486658546) node {$P_{4}$};
\draw [fill=ududff] (9.534085607183261,6.083288781597591) circle (2.5pt);
\draw[color=ududff] (9.697048588562623,6.431058127684059) node {$P_1$};
\draw [fill=uuuuuu] (7.552008002205401,7.158268508597804) circle (2.0pt);
\draw[color=uuuuuu] (7.932745708220286,7.453364469564665) node {$P_2$};
\draw [fill=uuuuuu] (6.731495463614648,4.431205257178962) circle (2.0pt);
\draw[color=uuuuuu] (7.075327485997843,4.271023760162134) node {$P_{7}$};
\draw [fill=uuuuuu] (4.74472582855669,7.329028416415277) circle (2.0pt);
\draw[color=uuuuuu] (4.371162323603981,7.684207837086092) node {$P_3$};
\draw [fill=uuuuuu] (5.436147184634907,4.393730185893523) circle (2.0pt);
\draw[color=uuuuuu] (5.030714802236631,4.28751257212795) node {$P_5$};
\draw[color=black] (6.671351592835345,6.653657089222578) node {$R$};
\draw[color=black] (5.33575782360423,6.5547242174276805) node {$L$};
\end{normalsize}
\end{tikzpicture}
         }
         &
         \scalebox{0.45}{
         \definecolor{qqccqq}{rgb}{0.,0.8,0.}
\definecolor{ffzzqq}{rgb}{1.,0.6,0.}
\definecolor{ccqqqq}{rgb}{0.8,0.,0.}
\definecolor{uuuuuu}{rgb}{0.26666666666666666,0.26666666666666666,0.26666666666666666}
\definecolor{qqzzff}{rgb}{0.,0.6,1.}
\definecolor{ududff}{rgb}{0.30196078431372547,0.30196078431372547,1.}
\begin{tikzpicture}[line cap=round,line join=round,>=triangle 45,x=1.0cm,y=1.0cm]
\clip(2.3,1.59) rectangle (10.,9.5);
\fill[line width=2.pt,color=ccqqqq,fill=ccqqqq,fill opacity=0.75] (7.581908305390389,7.257645439707362) -- (6.65703865500993,7.814012338764357) -- (6.205524916558751,3.818651599391875) -- (6.635136015598504,4.1109440823001995) -- cycle;
\fill[line width=2.pt,color=ffzzqq,fill=ffzzqq,fill opacity=0.75] (6.65703865500993,7.814012338764357) -- (5.6431653478231025,8.423920500118934) -- (5.886405446639656,4.044987434885154) -- (6.205524916558751,3.818651599391875) -- cycle;
\fill[line width=2.pt,color=qqccqq,fill=qqccqq,fill opacity=0.699999988079071] (5.6431653478231025,8.423920500118934) -- (4.659916294585864,7.689071207699526) -- (5.444460558602645,4.3584373334488475) -- (5.886405446639656,4.044987434885154) -- cycle;
\draw [line width=2.pt,color=qqzzff,domain=6.0:10.0] plot(\x,{(-30.411428277928287--5.633577322436402*\x)/1.6950178283450619});
\draw [line width=2.pt,color=qqzzff,domain=2.3:6.0] plot(\x,{(-37.05570026152593--5.726329989586373*\x)/-1.348860162003846});
\draw [line width=1.2pt,dash pattern=on 3pt off 3pt,domain=2.3:6.0] plot(\x,{(-39.2571923050674--6.423920500118934*\x)/-0.35683465217689747});
\draw [line width=1.2pt,dash pattern=on 3pt off 3pt,domain=6.0:10.0] plot(\x,{(-10.500859763233748--1.8186515993918748*\x)/0.20552491655875116});
\draw [line width=2.pt] (2.755372967829234,6.265675668755097)-- (5.6431653478231025,8.423920500118934);
\draw [line width=2.pt] (5.6431653478231025,8.423920500118934)-- (9.534085607183261,6.083288781597591);
\draw [line width=2.pt] (9.534085607183261,6.083288781597591)-- (6.205524916558751,3.818651599391875);
\draw [line width=2.pt] (6.205524916558751,3.818651599391875)-- (2.755372967829234,6.265675668755097);
\begin{normalsize}
\draw [fill=ududff] (6.,2.) circle (2.5pt);
\draw[color=ududff] (6.226153669758307,1.8224351832384242) node {$C$};
\draw [fill=ududff] (6.205524916558751,3.818651599391875) circle (2.5pt);
\draw[color=ududff] (6.292108917621571,3.45482756785423) node {$B$};
\draw [fill=ududff] (5.6431653478231025,8.423920500118934) circle (2.5pt);
\draw[color=ududff] (5.830422182578717,9.061023636231747) node {$B'$};
\draw [fill=ududff] (2.755372967829234,6.265675668755097) circle (2.5pt);
\draw[color=ududff] (2.5573930073641966,6.925722486658546) node {$P_{4}$};
\draw [fill=ududff] (9.534085607183261,6.083288781597591) circle (2.5pt);
\draw[color=ududff] (9.697048588562623,6.431058127684059) node {$P_1$};
\draw [fill=uuuuuu] (7.581908305390389,7.257645439707362) circle (2.0pt);
\draw[color=uuuuuu] (7.965723332151918,7.552297341359562) node {$P_2$};
\draw [fill=uuuuuu] (6.635136015598504,4.1109440823001995) circle (2.0pt);
\draw[color=uuuuuu] (6.9763946142029445,3.9412475208458093) node {$P_{7}$};
\draw [fill=uuuuuu] (4.659916294585864,7.689071207699526) circle (2.0pt);
\draw[color=uuuuuu] (4.2887182637749,8.04696170033405) node {$P_3$};
\draw [fill=uuuuuu] (5.444460558602645,4.3584373334488475) circle (2.0pt);
\draw[color=uuuuuu] (5.047203614202447,4.2545349481963175) node {$P_5$};
\draw [fill=uuuuuu] (6.65703865500993,7.814012338764357) circle (2.0pt);
\draw[color=uuuuuu] (7.017616644117486,8.08818373024859) node {$A$};
\draw [fill=uuuuuu] (5.886405446639656,4.044987434885154) circle (2.0pt);
\draw[color=uuuuuu] (5.673778468903463,3.908269896914177) node {$P_6$};
\draw[color=black] (6.869217336425139,6.505257781530232) node {$R$};
\draw[color=black] (6.011799114202696,7.478097687513389) node {$M$};
\draw[color=black] (5.203847327877701,6.752589961017475) node {$L$};
\end{normalsize}
\end{tikzpicture}
}
         &
         \scalebox{0.43}{
\definecolor{ffzzqq}{rgb}{1.,0.6,0.}
\definecolor{ccqqqq}{rgb}{0.8,0.,0.}
\definecolor{qqccqq}{rgb}{0.,0.8,0.}
\definecolor{qqzzff}{rgb}{0.,0.6,1.}
\definecolor{ududff}{rgb}{0.30196078431372547,0.30196078431372547,1.}
\begin{tikzpicture}[line cap=round,line join=round,>=triangle 45,x=1.0cm,y=1.0cm]
\clip(2.7,2.5) rectangle (9.8,10.);
\fill[line width=2.pt,color=qqccqq,fill=qqccqq,fill opacity=0.75] (5.258,9.561818181818186) -- (4.429561323598802,8.988801978397113) -- (5.490076665237584,4.944571234741959) -- (5.794711738953986,4.815450463403734) -- cycle;
\fill[line width=2.pt,color=ccqqqq,fill=ccqqqq,fill opacity=0.75] (6.38,5.007818181818187) -- (6.801245926127107,5.165327528109192) -- (8.14058486110408,8.784824827012187) -- (7.172,9.473818181818187) -- cycle;
\fill[line width=2.pt,color=ffzzqq,fill=ffzzqq,fill opacity=0.4000000059604645] (5.794711738953986,4.815450463403734) -- (6.38,5.007818181818187) -- (7.172,9.473818181818187) -- (5.258,9.561818181818186) -- cycle;
\draw [line width=2.pt,color=qqzzff,domain=6.0:9.8] plot(\x,{(-15.196909090909116--3.1078181818181863*\x)/1.15});
\draw [line width=2.pt,color=qqzzff,domain=2.7:6.0] plot(\x,{(-21.753970392177173--3.205388518718568*\x)/-0.8405464266219216});
\draw [line width=2.pt] (3.,8.)-- (5.258,9.561818181818186);
\draw [line width=2.pt] (3.,6.)-- (5.808,4.809818181818186);
\draw [line width=2.pt] (8.998,5.975818181818187)-- (8.91,5.953818181818186);
\draw [line width=2.pt] (8.91,5.953818181818186)-- (6.38,5.007818181818187);
\draw [line width=2.pt] (7.172,9.473818181818187)-- (9.306,7.955818181818186);
\draw [line width=1.2pt,dash pattern=on 2pt off 2pt,domain=6.0:9.8] plot(\x,{(-10.906909090909117--2.0078181818181866*\x)/0.38});
\draw [line width=1.2pt,dash pattern=on 2pt off 2pt,domain=2.7:6.0] plot(\x,{(-41.596909090909115--6.561818181818186*\x)/-0.742});
\begin{large}
\draw [fill=ududff] (6.,3.) circle (2.5pt);
\draw[color=ududff] (5.293385552812234,3.0523359275163826) node {$C$};
\draw [fill=ududff] (3.,8.) circle (2.5pt);
\draw[color=ududff] (3.1386258183450044,8.351880139314163) node {$A$};
\draw [fill=ududff] (3.,6.) circle (2.5pt);
\draw[color=ududff] (3.1386258183450044,6.352418403727454) node {$B$};
\draw [fill=ududff] (5.808,4.809818181818186) circle (2.5pt);
\draw[color=ududff] (5.302822545535621,4.391781167860876) node {$Z$};
\draw [fill=ududff] (7.172,9.473818181818187) circle (2.5pt);
\draw[color=ududff] (7.584031036299919,9.749562129238853) node {$F$};
\draw [fill=ududff] (9.306,7.955818181818186) circle (2.5pt);
\draw[color=ududff] (9.447607022866173,8.313055639594033) node {$D$};
\draw [fill=ududff] (8.91,5.953818181818186) circle (2.5pt);
\draw[color=ududff] (9.039949775804804,6.313593904007324) node {$E$};
\draw[color=black] (4.972,8.098818181818187) node {$L$};
\draw[color=black] (7.282,8.076818181818187) node {$R$};
\draw[color=black] (6.182,8.230818181818186) node {$M$};
\draw[color=black] (6.182,9.330818181818186) node {$\cdots$};
\draw[color=black] (6.082,5.330818181818186) node {$\cdots$};
\end{large}
\end{tikzpicture}
         }
         \\
         (a) & (b) & (c) 
    \end{tabular}
    \caption{The case of intersection $\mathcal{P}\cap S[C,\varepsilon_r,\varepsilon_\ell]$ has two colinear vertices with $C$, has two non-colinear vertices and has more than two vertices.}
    \label{fig:LMR_Cases}
\end{figure}

\noindent
In the 
next
section, we will partition
the domain of directions $\mathcal{D}_G\subseteq[0,2\pi]$ into intervals $[\theta_i,\theta_{i+1}]$, where only the left, and right area change (see Remark~\ref{remark:M_equal}) for every rotation $\theta \in [\theta_i,\theta_{i+1}]$.
This means that the area of the intersection $S(C,\theta,\phi)\cap \mathcal{P}$ in each interval is  
\begin{align*}
    Area(S(C,\theta,\phi)\cap \mathcal{P})= f_i(\theta) &= Area(L) +Area(M) + Area(R) \\
    &= A^L_\phi(\theta) + A^R_\phi(\theta) + Area(M) && \theta \in [\theta_i,\theta_{i+1}]
\end{align*}
Where $Area(M)$ is a constant, and it can be calculated either using the shoelace formula~\cite{Shoelace} or as the sum of its quadrilateral sections.
Now we can define properly the objective function of the area $f: \mathcal{D}_G \rightarrow \mathbb{R}$ as
\begin{align}
{\footnotesize
\label{eq:Func_whole}
    f(\theta) = 
    \begin{cases}
        f_1(\theta) = A^{L_1}_\phi(\theta) + A^{R_1}_\phi(\theta) + Area(M_1) \hspace{2cm} \theta \in [\theta_1,\theta_{2}]
        \\
        \hspace{2cm} \vdots
        \\
        f_i(\theta) = A^{L_i}_\phi(\theta) + A^{R_i}_\phi(\theta) + Area(M_i) \hspace{2cm}  \theta \in [\theta_i,\theta_{i+1}]
        \\
        \hspace{2cm} \vdots
        \\
        f_{Q-1}(\theta) = A^{L_{Q-1}}_\phi(\theta) + A^{R_{Q-1}}_\phi(\theta) + Area(M_{Q-1}) \hspace{0.5cm}  \theta \in [\theta_{Q-1},\theta_{Q}]
    \end{cases}
    }
\end{align}

\section{Partitioning P into finite LMR cells}

In every optimisation problem, the optimal value is obtained by the minimisation or maximisation of a given objective function~\cite{cont_opt} and in this problem, this objective function is the area of the intersection. One of the challenges of this problem is to express the area of the intersection as a function of rotations in a systematic way because the intersection can be in many shapes (see Figure~\ref{fig:Problem1}). In this section, we present a partition of the polygon into a sequence of quadrilaterals. Using them as a point of reference not only can we express the area of intersection in a systematic way but also we prove that there are finite independent sub-problems.
%
Let us now explain how to decompose an infinite set of intersections into a finite number of independent sub-problems.  
First, we partition the polygon into quadrilateral sections by a set of lines from a point $C$ to polygon vertices. Then every intersection $S(C,\theta,\phi)\cap \mathcal{P}$ can be written as a union of three unique convex sets $L$,$M$, and $R$ (Left, Middle, Right), where $L$ and $R$ are subsets of the polygon's sections. By defining an equivalence relation on the $L$, $M$, $R$ sets, we are partitioning all intersections  $S(C,\theta,\phi)\cap \mathcal{P}$ into finite families.
So, 
we can obtain the maximal intersection by selecting the maximum overall maximums in the equivalent classes.
%
More formally:

\begin{definition}
    A \textbf{partition} of a set $S$ is a collection of nonempty subsets of $S$ such that every element of $S$ is in exactly one of the subsets. The subsets are the \textbf{cells} of the partition.
\end{definition}

\begin{notation}[Counterclockwise Vertices' Angular Ordering]
    Let $\mathcal{P}=(P_1,\dots,P_n)$ be a polygon with $n$ vertices, a point $C$ outside of $\mathcal{P}$, and the semi-lines of $CP_i$ that extend from $C$. We will denote with $\{P_{k_i}\}_{i=1}^n$
    the strictly increasing sequence of the vertices such that $P_{k_i}<P_{k_{i+1}}$ if $\hat{CP_{i}} < \hat{CP_{i+1}}$ (see Figure~\ref{fig:LMP_partition}(a)).
\end{notation}

\begin{definition}
Let $\mathcal{P}=(P_1,\dots,P_n)$ be a polygon with $n$ vertices, a point $C$ outside of $\mathcal{P}$, and the semi-lines of $CP_i$ that extend from $C$. We will call the sequence $\{S_i\}_{i=1}^{m}$, $m\leq n-1$, \textbf{vertex partitioning of $\mathcal{P}$ from $C$} and a set $S_i$ a \textbf{section} of $\mathcal{P}$ 
where:
    $S_i = S[C,CP_{k_{i+1}},CP_{k_i}]\cap \mathcal{P}.$
\end{definition}

Notice that the sequence $S_i$ partitions the polygon $\mathcal{P}$ using the angular position of the vertices of $\mathcal{P}$ from point $C$. 
Sorting the semi-lines $CP_i$ in a strictly increasing sequence means that we exclude semi-lines where $CP_i$ coincides with $CP_{i+1}$, which means the number of sections is at most $n-1$.

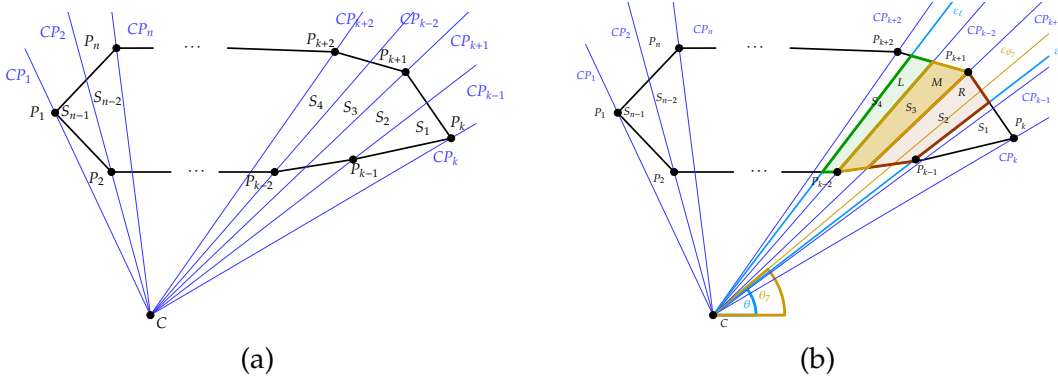
\begin{figure}[ht]
    \centering
    \begin{tabular}{c c}
         \scalebox{0.53}{
    \definecolor{ttttff}{rgb}{0.2,0.2,1.}
    \begin{tikzpicture}[line cap=round,line join=round,>=triangle 45,x=0.5cm,y=0.5cm]
    \clip(0,-5.) rectangle (25.5,11);
    \draw [line width=1.2pt] (5.644745971611137,9.197094032978754)-- (2.597750273663621,5.9634325692927925);
    \draw [line width=1.2pt] (2.597750273663621,5.9634325692927925)-- (5.39724,3.);
    \draw [line width=1.2pt] (5.39724,3.)-- (8.34047,3.);
    \draw [line width=1.2pt] (5.644745971611137,9.197094032978754)-- (7.46068,9.2);
    \draw [line width=1.2pt] (16.48973396390154,9.00940866879776)-- (20.,8.);
    \draw [line width=1.2pt] (20.,8.)-- (22.263188754469994,4.688759191468466);
    \draw [line width=1.2pt] (22.263188754469994,4.688759191468466)-- (17.393677637476703,3.6256969053624473);
    \draw [line width=1.2pt] (17.393677637476703,3.6256969053624473)-- (13.5,3.);
    \draw [line width=0.4pt,color=ttttff,domain=7.330195156352894:23.5] plot(\x,{(-78.1817708587774--7.162310028868053*\x)/6.169804843647106});
    \draw [line width=0.4pt,color=ttttff,domain=7.330195156352894:23.5] plot(\x,{(-98.97494476345875--7.7880069342305*\x)/10.06348248112381});
    \draw [line width=0.4pt,color=ttttff,domain=7.330195156352894:23.5] plot(\x,{(-127.03581374192015--8.85106922033652*\x)/14.9329935981171});
    \draw [line width=0.4pt,color=ttttff,domain=7.330195156352894:23.5] plot(\x,{(-141.8877618281842--12.162310028868053*\x)/12.669804843647107});
    \draw [line width=0.4pt,color=ttttff,domain=7.330195156352894:23.5] plot(\x,{(-134.67610883693862--13.171718697665813*\x)/9.159538807548646});
    \draw [line width=0.4pt,color=ttttff,domain=1.1:7.330195156352894] plot(\x,{(-90.91167690111254--13.359404061846806*\x)/-1.6854491847417572});
    \draw [line width=0.4pt,color=ttttff,domain=1.1:7.330195156352894] plot(\x,{(-54.52576655123194--10.125742598160844*\x)/-4.732444882689273});
    \draw [line width=0.4pt,color=ttttff,domain=1.1:7.330195156352894] plot(\x,{(-44.45557164926649--7.162310028868053*\x)/-1.932955156352894});
    \draw [line width=1.2pt] (16.48973396390154,9.00940866879776)-- (11.16442,9.2);
    \draw [line width=1.2pt] (13.5,3.)-- (10.7465,3.);
    \begin{normalsize}
    \draw [fill=black] (5.644745971611137,9.197094032978754) circle (2.5pt);
    \draw[color=black] (4.449504051706032,9.497290633402772) node {$P_n$};
    \draw [fill=black] (2.597750273663621,5.9634325692927925) circle (2.5pt);
    \draw[color=black] (1.731304534625714,5.942722034137706) node {$P_1$};
    \draw [fill=black] (5.39724,3.) circle (2.5pt);
    \draw[color=black] (4.658596322250671,2.6669431289326457) node {$P_2$};
    \draw [fill=black] (22.263188754469994,4.688759191468466) circle (2.5pt);
    \draw[color=black] (22.605682877332253,5.419991357775197) node {$P_{k}$};
    \draw [fill=black] (17.393677637476703,3.6256969053624473) circle (2.5pt);
    \draw[color=black] (17.970804213592736,3.1) node {$P_{k-1}$};
    \draw [fill=black] (20.,8.) circle (2.5pt);
    \draw[color=black] (19.4,8.8) node {$P_{k+1}$};
    \draw [fill=black] (13.5,3.) circle (2.5pt);
    \draw[color=black] (12.743497449976743,2.5) node {$P_{k-2}$};
    \draw [fill=black] (16.48973396390154,9.00940866879776) circle (2.5pt);
    \draw[color=black] (15.79656065480554,9.60718733280282) node {$P_{k+2}$};
    \draw [fill=black] (7.330195156352894,-4.162310028868053) circle (2.5pt);
    \draw[color=black] (7.882102159813868,-4.564164560748736) node {$C$};
    \draw[color=ttttff] (20.7,10.6) node {$CP_{k-2}$};
    \draw[color=ttttff] (24,7) node {$CP_{k-1}$};
    \draw[color=ttttff] (22,3.7) node {$CP_{k}$};
    \draw[color=ttttff] (23.2,9.5) node {$CP_{k+1}$};
    \draw[color=ttttff] (17.5,10.6) node {$CP_{k+2}$};
    \draw[color=ttttff] (6.8540651629693885,10.132094417175145) node {$CP_n$};
    \draw[color=ttttff] (0.8,7.806337975962648) node {$CP_1$};
    \draw[color=ttttff] (2.5,9.969968893782633) node {$CP_2$};
    
    \node[text width=0.5cm] at (21,5.2) 
    {$S_1$};
    \node[text width=0.5cm] at (19,5.7) 
    {$S_2$};
    \node[text width=0.5cm] at (17.4,6.2) 
    {$S_3$};
    \node[text width=0.5cm] at (15.7,6.5) 
    {$S_4$};
    \node[text width=0.7cm] at (5.2,6.7) 
    {$S_{n-2}$};
    \node[text width=0.7cm] at (3.6,6) 
    {$S_{n-1}$};
    \end{normalsize}
    
    \node[text width=0.5cm] at (9.5,9.197094032978754) 
    {$\cdots$};
    
    \node[text width=0.5cm] at (9.7,3) 
    {$\cdots$};
    
    \end{tikzpicture}
    }
         &
         \scalebox{0.53}{
         \definecolor{zzttqq}{rgb}{0.6,0.2,0.}
    \definecolor{qqzzqq}{rgb}{0.,0.6,0.}
    \definecolor{cczzqq}{rgb}{0.8,0.6,0.}
    \definecolor{qqzzff}{rgb}{0.,0.6,1.}
    \definecolor{ttttff}{rgb}{0.2,0.2,1.}
    \begin{tikzpicture}[line cap=round,line join=round,>=triangle 45,x=0.5cm,y=0.5cm]
    \clip(0.5,-5.) rectangle (25.5,11.5);
    \fill[line width=2.pt,color=qqzzqq,fill=qqzzqq,fill opacity=0.10000000149011612] (12.757737917602592,3.) -- (13.5,3.) -- (18.242488318970842,8.505388340399392) -- (17.164567942824316,8.815354069710724) -- cycle;
    \fill[line width=2.pt,color=zzttqq,fill=zzttqq,fill opacity=0.10000000149011612] (20.,8.) -- (15.015365010907576,3.2150268400950295) -- (17.229891739007773,3.5292631338076417) -- (21.03502768850645,6.485665009878843) -- cycle;
    \fill[line width=2.pt,color=cczzqq,fill=cczzqq,fill opacity=0.3499999940395355] (18.242488318970842,8.505388340399392) -- (20.,8.) -- (15.015365010907576,3.2150268400950295) -- (13.5,3.) -- cycle;
    \draw [shift={(7.330195156352894,-4.162310028868053)},line width=2.pt,color=qqzzff] (0,0) -- (1.6988075267329586E-7:2.1349851671442264) arc (1.6988075267329586E-7:37.84543447944551:2.1349851671442264) -- cycle;
    \draw [shift={(7.330195156352894,-4.162310028868053)},line width=2.pt,color=cczzqq] (0,0) -- (1.6988075267329586E-7:3.558308611907044) arc (1.6988075267329586E-7:40.18546669720131:3.558308611907044) -- cycle;
    \draw [line width=1.2pt] (5.644745971611137,9.197094032978754)-- (2.597750273663621,5.9634325692927925);
    \draw [line width=1.2pt] (2.597750273663621,5.9634325692927925)-- (5.39724,3.);
    \draw [line width=1.2pt] (5.39724,3.)-- (8.34047,3.);
    \draw [line width=1.2pt] (5.644745971611137,9.197094032978754)-- (7.46068,9.2);
    \draw [line width=1.2pt] (16.48973396390154,9.00940866879776)-- (20.,8.);
    \draw [line width=1.2pt] (20.,8.)-- (22.263188754469994,4.688759191468466);
    \draw [line width=1.2pt] (22.263188754469994,4.688759191468466)-- (17.76649646830244,3.6054061214639126);
    \draw [line width=1.2pt] (17.76649646830244,3.6054061214639126)-- (13.5,3.);
    \draw [line width=0.4pt,color=ttttff,domain=7.330195156352894:24.0] plot(\x,{(-78.1817708587774--7.162310028868053*\x)/6.169804843647106});
    \draw [line width=0.4pt,color=ttttff,domain=7.330195156352894:24.0] plot(\x,{(-100.37799691610392--7.767716150331966*\x)/10.436301311949546});
    \draw [line width=0.4pt,color=ttttff,domain=7.330195156352894:24.0] plot(\x,{(-127.03581374192015--8.85106922033652*\x)/14.9329935981171});
    \draw [line width=0.4pt,color=ttttff,domain=7.330195156352894:24.0] plot(\x,{(-141.8877618281842--12.162310028868053*\x)/12.669804843647107});
    \draw [line width=0.4pt,color=ttttff,domain=7.330195156352894:24.0] plot(\x,{(-134.67610883693862--13.171718697665813*\x)/9.159538807548646});
    \draw [line width=0.4pt,color=ttttff,domain=1.0:7.330195156352894] plot(\x,{(-90.91167690111254--13.359404061846806*\x)/-1.6854491847417572});
    \draw [line width=0.4pt,color=ttttff,domain=1.0:7.330195156352894] plot(\x,{(-54.52576655123194--10.125742598160844*\x)/-4.732444882689273});
    \draw [line width=0.4pt,color=ttttff,domain=1.0:7.330195156352894] plot(\x,{(-44.45557164926649--7.162310028868053*\x)/-1.932955156352894});
    \draw [line width=1.2pt] (16.48973396390154,9.00940866879776)-- (11.16442,9.2);
    \draw [line width=1.2pt] (13.5,3.)-- (10.7465,3.);
    \draw [line width=1.2pt,color=qqzzff,domain=7.330195156352894:24.0] plot(\x,{(-141.0444757196725--11.116862447388947*\x)/14.308329759418939});
    \draw [line width=1.2pt,color=qqzzff,domain=7.330195156352894:24.0] plot(\x,{(-151.40815035556946--14.441332790580873*\x)/10.943529522513927});
    \draw [line width=0.4pt,color=cczzqq,domain=7.330195156352894:24.0] plot(\x,{(-126.89940501435674--10.3522401789408*\x)/12.256526746797096});
    \draw [line width=2.pt,color=qqzzqq] (12.757737917602592,3.)-- (13.5,3.);
    \draw [line width=2.pt,color=qqzzqq] (13.5,3.)-- (18.242488318970842,8.505388340399392);
    \draw [line width=2.pt,color=qqzzqq] (18.242488318970842,8.505388340399392)-- (17.164567942824316,8.815354069710724);
    \draw [line width=2.pt,color=qqzzqq] (17.164567942824316,8.815354069710724)-- (12.757737917602592,3.);
    \draw [line width=2.pt,color=zzttqq] (20.,8.)-- (15.015365010907576,3.2150268400950295);
    \draw [line width=2.pt,color=zzttqq] (15.015365010907576,3.2150268400950295)-- (17.229891739007773,3.5292631338076417);
    \draw [line width=2.pt,color=zzttqq] (17.229891739007773,3.5292631338076417)-- (21.03502768850645,6.485665009878843);
    \draw [line width=2.pt,color=zzttqq] (21.03502768850645,6.485665009878843)-- (20.,8.);
    \draw [line width=2.pt,color=cczzqq] (18.242488318970842,8.505388340399392)-- (20.,8.);
    \draw [line width=2.pt,color=cczzqq] (20.,8.)-- (15.015365010907576,3.2150268400950295);
    \draw [line width=2.pt,color=cczzqq] (15.015365010907576,3.2150268400950295)-- (13.5,3.);
    \draw [line width=2.pt,color=cczzqq] (13.5,3.)-- (18.242488318970842,8.505388340399392);
       \begin{scriptsize}
    \draw [fill=black] (5.644745971611137,9.197094032978754) circle (2.5pt);
    \draw[color=black] (4.449504051706032,9.497290633402772) node {$P_n$};
    \draw [fill=black] (2.597750273663621,5.9634325692927925) circle (2.5pt);
    \draw[color=black] (1.731304534625714,5.942722034137706) node {$P_1$};
    \draw [fill=black] (5.39724,3.) circle (2.5pt);
    \draw[color=black] (4.658596322250671,2.6669431289326457) node {$P_2$};
    \draw [fill=black] (22.263188754469994,4.688759191468466) circle (2.5pt);
    \draw[color=black] (22.605682877332253,5.419991357775197) node {$P_{k}$};
    \draw [fill=black] (17.393677637476703,3.6256969053624473) circle (2.5pt);
    \draw[color=black] (17.970804213592736,3.1) node {$P_{k-1}$};
    \draw [fill=black] (20.,8.) circle (2.5pt);
    \draw[color=black] (19.4,8.8) node {$P_{k+1}$};
    \draw [fill=black] (13.5,3.) circle (2.5pt);
    \draw[color=black] (12.743497449976743,2.5) node {$P_{k-2}$};
    \draw [fill=black] (16.48973396390154,9.00940866879776) circle (2.5pt);
    \draw[color=black] (15.79656065480554,9.60718733280282) node {$P_{k+2}$};
    \draw [fill=black] (7.330195156352894,-4.162310028868053) circle (2.5pt);
    \draw[color=black] (7.882102159813868,-4.564164560748736) node {$C$};
    \draw[color=ttttff] (20.68,10.0396663172976344) node {$CP_{k-2}$};
    \draw[color=ttttff] (23.4,6.7) node {$CP_{k-1}$};
    \draw[color=ttttff] (22,3.7) node {$CP_{k}$};
    \draw[color=ttttff] (24,10.6) node {$CP_{k+1}$};
    \draw[color=ttttff] (16,10.7) node {$CP_{k+2}$};
    \draw[color=ttttff] (6.8540651629693885,10.132094417175145) node {$CP_n$};
    \draw[color=ttttff] (1.055555581408433,7.806337975962648) node {$CP_1$};
    \draw[color=ttttff] (2.773755098488749,9.969968893782633) node {$CP_2$};
    
    \node[text width=0.5cm] at (21,5.2) 
    {$S_1$};
    \node[text width=0.5cm] at (19,5.7) 
    {$S_2$};
    \node[text width=0.5cm] at (17.4,6.2) 
    {$S_3$};
    \node[text width=0.5cm] at (15.7,6.5) 
    {$S_4$};
    \node[text width=0.7cm] at (5.2,6.7) 
    {$S_{n-2}$};
    \node[text width=0.7cm] at (3.6,6) 
    {$S_{n-1}$};
    
    \node[text width=0.5cm] at (17,7.5) 
    {$L$};
    \node[text width=0.5cm] at (18.7,7.5) 
    {$M$};
    \node[text width=0.5cm] at (20,7) 
    {$R$};
    
    \draw[color=cczzqq] (22,9) node {$\varepsilon_{\theta_7}$};
    \draw[color=qqzzff] (9.020846456185891,-3.604883779283618) node {$\theta$};
    \draw[color=qqzzff] (19.5,11) node {$\varepsilon_\ell$};
    \draw[color=qqzzff] (24.5,9) node {$\varepsilon_r$};
    \draw[color=cczzqq] (9.892632066103117,-3.195678288913595) node {$\theta_7$};
    \end{scriptsize}
    
    \node[text width=0.5cm] at (9.5,9.197094032978754) 
    {$\cdots$};
    
    \node[text width=0.5cm] at (9.7,3) 
    {$\cdots$};
    
    \end{tikzpicture}
    } 
         \\
         (a) & (b) 
    \end{tabular}
    \caption{(a) The vertex partitioning of $\mathcal{P}$ from $C$. The vertices of $\mathcal{P}$ are sorted counterclockwise, and the sequence $S_1,\ldots,S_{n-1}$ is sorted from right to left.
    (b) The LMR partition with $L\subseteq S_4$, $R\subseteq S_2$, and $M=S_3$ is valid for every $\theta \in [\hat{CP}_{k-1},\theta_7] $, where $\theta_7 = \hat{CP}_{k+2}-\phi$. When $\theta=\theta_7$ then $\varepsilon_r$ will cross $CP_{k+2}$.
    }
    \label{fig:LMP_partition}
\end{figure}

We can partition the intersection into three sets, Left, Middle, and Right; where during a ``small rotation '' only the Left and Right are quadrilaterals and change while the middle remains constant. In Lemma~\ref{lemma:LMR_exist}, we show that these can be expressed uniquely if defined as in Definition~\ref{def:LMR_ex} and also illustrate an example in Figure~\ref{fig:LMR_Cases}.

\begin{definition}
\label{def:LMR_ex}
    Let $S(C,\theta,\phi)$ be a sector that either fully or partially intersects a polygon $\mathcal{P}$, $K=\mathcal{P}\cap S(C,\theta,\phi)$, and $\{S_i\}_{i=1}^m$ is the vertex partitioning of $\mathcal{P}$ from $C$. Let us define three sets $L$, $M$, $R$ (Left, Middle, Right) for the intersection $K$:
        
        \noindent
        $\bullet$ $L=\emptyset$ or $L= K\cap S_i$  for $i=\max\left\{q \in \{1,\ldots,m\}:S_q\cap K \neq \emptyset \right\}$ or $L=K$ if $K\subseteq S_i$, $i\in \{1,\ldots,m\}$;
         \\ $\bullet$ 
        $R=\emptyset$ or $R= K\cap S_j$, for $i\neq j=\min\left\{q\in \{1,\ldots,m\}:S_q\cap K \neq \emptyset \right\}$; 
        \\ $\bullet$ 
         $M=K\setminus(L\cup R)$.
\end{definition}

\begin{lemma}
\label{lemma:LMR_exist}
    Every intersection $K=\mathcal{P}\cap S(C,\theta,\phi)$ is a union of three unique sets $L,M,R$ as defined in Definition~\ref{def:LMR_ex}.
\end{lemma}

\begin{proof}
    We only need to consider how many elements set $H=\left\{q \in \{1,\ldots,m\}: S_q\cap K \neq \emptyset \right\}$ contains. 
    If the cardinality of $|H|=1$, then $K\subseteq S_i$, for $i\in \{1,\ldots,m\}$, and the triplet $(L,M,R) = (K\cap S_i,\emptyset,\emptyset)$. If $|H|=2$, then $i = max(H)$, and $j=min(H)$, so the triplet $(L,M,R) = (K\cap S_i,\emptyset,K \cap S_j)$. If $|H|>2$, then $i = max(H)$, and $j=min(H)$, so the triplet $(L,M,R) = (K\cap S_i,K\setminus(L\cup R),K \cap S_j)$.

    It is apparent that $K= L\cup M \cup R$, and the uniqueness of the sets stems from the uniqueness of the minimum and the maximum element of $H$. Finally, notice that the only ambiguous case is when $|H|=1$, but then we select $L=K$.
\end{proof}

The result of Lemma~\ref{lemma:LMR_exist} means that there is a bijection between an intersection $\mathcal{P}\cap S(C,\theta,\phi)$ and its decomposition $L$, $M$, $R$, so an equivalence relation on these sets not only partitions them but partitions the intersections as well. 
The intersections expressed their decomposition $L$,$M$,$R$ share the same branch of equation~(\ref{eq:Func_whole}) if both their left sets are subsets of the same section and at the same time, both their right sets are subsets of the same section.

\begin{definition}
\label{def:LMR_Famillies}
    Let two intersections $K_1=S(C,\theta_1,\phi) \cap \mathcal{P}=L_1\cup M_1 \cup R_1$, $K_2=S(C,\theta_2,\phi) \cap \mathcal{P}=L_2\cup M_2 \cup R_2 $, and  $\{S_i \}_{i=1}^m$ be the vertex partitioning of $\mathcal{P}$ from $C$. We define the relation $\mathcal{LMR}$, and we will say that $K_1$ and $K_2$ are \textbf{$\mathcal{LMR}$ related} if and only if the following statements are both true, for $j<i\in \{1,\ldots,n\}$:
    \begin{itemize}
        \item $L_1\subseteq L_2 \subseteq  S_i$ or $L_2\subseteq L_1 \subseteq  S_i$ or $L_1=L_2=\emptyset$
        \item $R_1\subseteq R_2 \subseteq  S_j$ or $R_2\subseteq R_1 \subseteq  S_j$ or $R_1=R_2=\emptyset$
    \end{itemize}
\end{definition}

\begin{remark}
\label{remark:M_equal}
    Notice that if two intersections $L_1,M_1,R_1$, and $L_2,M_2,R_2$ are $\mathcal{LMR}$ related, then $M_1=M_2$ because either $M_1=M_2=\emptyset$ or if we consider
$
        H=\left\{q \in \{1,\ldots,m\}: S_q\cap K \neq \emptyset \right\}, 
$
    $i = \max(H)$, and $j=\min(H)$ then $M_1=M_2= \cup_{k=j+1}^{i-1}S_k$.
\end{remark}

\begin{lemma}
\label{lemma:LMR_Families_eq}
    The relation $\mathcal{LMR}$ is a relation of equivalence. 
\end{lemma}

\begin{proof}
    Let three partitions $L_1,M_1,R_1$, $L_2,M_2,R_2$, $L_3,M_3,R_3$  of three different intersection, and $\{S_i \}_{i=1}^m$ as Lemma~\ref{lemma:LMR_Families_eq} assumes. We will denote a partition as a pair $(L_i,R_i)$, $i\in \{1,2,3\}$. $(L_i,R_i)\sim (L_j,R_j)$ denotes that the partitions $(L_i,R_i)$, and $(L_j,R_j)$ belong in the same LMR family. From definition~\ref{def:LMR_Famillies} it is directly induced that the reflective ($(L_i,R_i)\sim (L_i,R_i)$) and the symmetric ( if  $(L_i,R_i)\sim (L_j,R_j)$ then $(L_j,R_j)\sim (L_i,R_i)$) properties are true. All we need to is to prove the transitive property. If $(L_1,R_1)\sim(L_2,R_2)$ and $(L_2,R_2)\sim(L_3,R_3)$, then we know that both $L_1$, and $L_3$ are subsets of a section $S_i$, which lead to the fact that either $L_1\subseteq L_3$ or $L_3\subseteq L_1$  or $L_1=L_3=\emptyset$. By the same logic both $R_1$, and $R_3$ are subsets of a section $S_j$, which lead to the fact that either $R_1\subseteq R_3$ or $R_3\subseteq R_1$  or $R_1=R_3=\emptyset$ which leads to $(L_1,R_1)\sim(L_3,R_3)$
\end{proof}
The intersections $S(C,\theta_1,\phi)\cap \mathcal{P}$, and $S(C,\theta_2,\phi)\cap \mathcal{P}$ are not $\mathcal{LMR}$ related if during the rotation from $\theta_1$ to $\theta_2$ either the left or the right semi-line of the sector $S(C,\theta,\phi)$ crosses one of the $CP_i$ lines of polygon $\mathcal{P}$, $i\in \{1.\ldots,n\}$. In other words, if in an interval $[\theta_s,\theta_f]$ contains an angle of rotation $\theta$ of the form $CP_i-\phi$ or $CP_i$ then the intersections $S(C,\theta_s,\phi)\cap \mathcal{P}$, and $S(C,\theta_f,\phi)\cap \mathcal{P}$ are not $\mathcal{LMR}$ related. Also, it is known~\cite{fraleigh_algebra} that an equivalence relation on a set $S$ yields a partition of $S$.

\begin{corollary}
    The equivalence relation $\mathcal{LMR}$ partitions naturally the domain of rotations into intervals $[\theta_i,\theta_{i+1}]$, where the sequence $\{\theta_i\}_{i=1}^Q$ is the merged sorted list of the two strictly increasing sequences of angles $\{\hat{CP_{k_i}}\}_{i=1}^n$, and $\{\hat{CP_{k_i}}-\phi\}_{i=1}^n$.
\end{corollary}

\begin{corollary}
The number of $\mathcal{LMR}$ cells $Q$ is at most $2n$.
\end{corollary}

\noindent
If a sector $S(C,\theta_0,\phi)$ intersects either fully or partially a polygon $\mathcal{P}=(P_1,\ldots,P_n)$ then from Lemma~\ref{lemma:LMR_exist} there exists a partition $L$,$M$,$R$ of $\mathcal{P} \cap S(C,\theta_0,\phi)$. The partition of $L$,$M$,$R$ belongs to a unique $\mathcal{LMR}$ cell over the interval of rotations $[\theta_i.\theta_{i+1}]$, $i\in \{1,...,Q-1\}$, and from Remark~\ref{remark:M_equal}, the change of the intersection area is equal to the change in the sum of the two quadrilaterals $L$ and $R$ for every rotation $\theta \in [\theta_i.\theta_{i+1}]$ as the area $M$ remains constant.

In the following section, we study the area of the intersection when it is a quadrilateral as a function of rotations. We will present a rotational sweep algorithm that approximates the maximum intersection by obtaining the maximum of all the approximated local maximums in the intervals $[\theta_i,\theta_{i+1}]$. Finally, we provide an analysis of the intersection for $\mathcal{LMR}$ cells.

\section{Maximum Intersection Algorithm}
\label{sec:Algo_Tech}
First, we need to compute the sections of the polygon  $\mathcal{P}$. We can do that using a rotational sweep on the vertices of $\mathcal{P}$ from the centre of the sector $C$, to compute all the $CP_i$ lines $i\in\{1,\ldots,n \}$, and their derivatives $\hat{CP}_i$.
Then we can compute also compute $\hat{CP}_i-\phi$ and merge them with the ordered list $\hat{CP}_i$ to create the sequence of $\{\theta_i \}_{i=1}^Q$ that make up the intervals $[\theta_i,\theta_{i+1}]$ of the independent problems.

If we consider that the intersection $K = \mathcal{P} \cap S(C,\theta,\phi)$ is a quadrilateral then we need to know the edges of the polygon that contribute to $K$ to be able to evaluate equation~(\ref{eq:General_Case_Alt}). We can identify the upper and lower edges of each section by using an algorithm that goes through the upper and lower hull of $\mathcal{P}$ from centre $C$  using the counter-clockwise order $(P_1,\ldots,P_n) $ and the order $(P_{k_1},\ldots,P_{k_n})$.

\begin{algorithm}[H]
{\footnotesize
\caption{\textit{Identifying the Lower and Upper Edge of Each Section}}
\label{Alg:Upper_Lower}
\begin{algorithmic}[1]
  \Require A polygon $(P_i)_{i=1}^n$, and a rearrangement $(P_{k_i})_{i=1}^n$ ordered by their $\hat{CP_{k_i}}$ values
  \Ensure Two lists of the upper and lower edge of each section of the polygon.
  \State \textbf{Set} upper\_Edges = $\{ P_{k_1}, P_{k_1+1} \}$, lower\_Edges = $\{ P_{k_1}, P_{k_1-1} \}$
  \State \textbf{Set} cu = $k_1+1$ //The index of the left vertex of the current Upper edge
  \For{$j=2$ to $n$}
  \State //Check if the change is in the upper or lower edge
  \If{$P_{k_j}$ \textbf{coincides with} $P_{cu}$}
  \State \textbf{Set} upper\_Edges = upper\_Edges~$\cup \{ P_{k_j+1} \}$,  cu = $k_j+1$
  \Else
  \State  \textbf{Set} lower\_Edges = lower\_Edges~$\cup \{ P_{k_j-1} \}$
  \EndIf
  \EndFor
  \State \Return upper\_Edges \textbf{and} lower\_Edges
\end{algorithmic}
}
\end{algorithm}

\noindent
Now we are ready to present an algorithm that $\varepsilon$ approximates the maximum intersection with accuracy $\varepsilon>1$, $|apx-opt|<10^{-\varepsilon}$.

\begin{algorithm}[H]
{\footnotesize
\caption{\textit{Maximum Intersection of a Polygon and a Sector
}}
\label{Alg:Max_Inter_Dir_theta}
\begin{algorithmic}[1]
  \Require The set P of the polygon's points, the centre $C\in \mathbb{R}^2$, the angle of the sector $\phi \in (0,\pi)$
  \Ensure The direction $\theta_{max}$ that yields to the intersection with the maximum area
  \State \textbf{Compute } $\hat{CP}_{k_1},\ldots,\hat{CP}_{k_n}$ by using $(P_1,\ldots,P_n,C)$
  \State \textbf{Compute }$\theta_1,\ldots,\theta_{2k_n}$
  by sorting the vlaues $(\hat{CP}_{k_1},\ldots,\hat{CP}_{k_n},\hat{CP}_{k_1}-\phi,\ldots,\hat{CP}_{k_n}-\phi)$
  \State \textbf{Compute} the upper and lower edges for each $LMR$
  using \textbf{Algorithm~\ref{Alg:Upper_Lower}}
  \For{$i=1$ to $2k_n-1$}
  \State //For every partition LMR find the points of the local maximum area from equation~\ref{eq:Func_whole}
  \State \textbf{Find } $r_i = \argmax_{r \in [0,\theta_{i+1}-\theta_i]}{f_i(r) }$ 
  \State  $local\_max_i = \theta_i + r_i$
  \EndFor
  \State $\theta_{max} = \max(local\_max_i)$
  \State \Return $\theta_{max}$
\end{algorithmic}
}
\end{algorithm}

\noindent
To find the maximum in line~6 then we examine all the local extreme points where $df_i/d\theta=0$ plus the values $f(\theta_i)$, and $f(\theta_{i+1})$. 

\begin{theorem}
Given a convex polygon $\mathcal{P}=(P_1,\ldots,P_n)$ with $n$ vertices, and a sector $S(C,\theta,\phi)$ where $\theta \in [0,2\pi]$, then Algorithm~\ref{Alg:Max_Inter_Dir_theta} approximates up to $\varepsilon$ accuracy the direction $\theta_{max}$ such that the area of $S(C,\theta_{max},\phi)\cap \mathcal{P}$ is maximised, in time $\mathcal{O}(n(\log{n}+\log{\varepsilon}/\phi))$.
\end{theorem}

\begin{proof}
Let $\{ \theta_i \}_{i=1}^{Q}$ be the sequence of domains where an LMR does not change.
The area of intersection $\mathcal{P}\cap S(\theta,\phi)$ for each cell of the partition, $i\in \{1,\ldots,Q\}$, is given from equation~(\ref{eq:Func_whole}) 
\begin{align*}
    f(\theta) = f_i(\theta)  \hspace{2cm} \theta \in [\theta_i,\theta_{i+1}], \hspace{2cm}
    i \in \{ 1,\ldots,Q\}
\end{align*}
We prove that Algorithm~\ref{Alg:Max_Inter_Dir_theta} returns a value $\theta^*$ that maximises $f$, that is $f(\theta^*)\geq f_A(\theta)$, $\theta \in \mathcal{D}_G$.
But first, we need to guarantee that such $\theta^*$ exists.

\begin{lemma}
\label{lemma:Exstance_Max}
There is at least one point $\theta^*\in \mathcal{D}_G$, so the function $f$ has a global maximum.
\end{lemma}

\begin{proof}
Let us consider the piecewise function  $f =f_i(\theta)$, $\theta\in[\theta_{i},\theta_{i+1}], i \in \{ 1,\ldots,Q\}$ which is continuous since each branch $f_i$ is continuous in the interval $[\theta_i,\theta_{i+1}]$. 
Also the domain of $f$ is compact because $\mathcal{D}_G=\bigcup_{i=1}^{Q}[\theta_i,\theta_{i+1}]$, this means $\mathcal{D}_G$ is the finite union of closed and bounded subsets of $\mathbb{R}$. Hence $f: \mathcal{D}_G\rightarrow \mathbb{R}$ is a continuous function defined in the compact set $\mathcal{D}_G$. So from the extreme value Theorem~\cite{spivak2019calculus}, there is at least one point $\theta^*\in \mathcal{D}_G$ such that $f(\theta^*)=\max{(f)}$.
\hspace{6.2cm}
\textit{End of proof of Lemma~\ref{lemma:Exstance_Max}}
\end{proof}

The sequence $\{ \theta_i \}_{i=1}^Q$ partitions the domain $\mathcal{D}_G$, if we find the local maximum of each partition, then the maximum of the local maximums should be the global maximum. At lines 4-8 this is what the Algorithm~\ref{Alg:Max_Inter_Dir_theta} does.

To find a local maximum, that is a maximum in each partition we search in $Z_i=\left\{r \in [\theta_i,\theta_{i+1}] : \frac{df_i}{dr}(r)=0 \right\}\cup\{\theta_i\}\cup\{\theta_{i+1}\}$.
Each $f_i$ is derivable because it is of the form of equation~(\ref{eq:Func_whole}) where $A^{L_i}_{\phi}$ is function~(\ref{eq:General_Case_Alt}) which is derivable, and $A^{R_i}_{\phi}$ is also function~(\ref{eq:General_Case_Alt}) with different arguments.
There are three cases for function $f_i$,
notice that when the sector intersects partially the polygon $\mathcal{P}$, then we 
and we need to be able to find the maximum of function $f$ in each case:

\begin{itemize}
    \item {\bf Case 1} If $S(C,\theta,\phi)$ contains $\mathcal{P})$ (i.e. $\phi>\theta_{Q}-\theta_1$) so the maximum intersection is the polygon itself.
    In this case, we can return $\theta_1$.
    \item {\bf  Case 2:} If $\mathcal{P}\cap S(C,\theta,\phi)$ is a subset of a section $S_i$, $in\in\{1,\ldots,n\}$ then from Theorem~\ref{Theor:Inter_Two_Sectors} we can approximate the maximum up to $\varepsilon$ accuracy.
    \item {\bf Case 3:} There are three partitions of Left Middle Right
    \begin{align*}
    f_i(\theta) = A^{L_i}_\phi(\theta) + A^{R_i}_\phi(\theta) + Area(M_i) &&  \theta \in [\theta_i,\theta_{i+1}]
    \end{align*}
    In this case we can still apply the same technique as in Theorem~\ref{Theor:Inter_Two_Sectors}. We rewrite functions  $A^{L_i}_\phi(\theta) + A^{R_i}_\phi(\theta)$ as four functions $A_\theta(\phi)$ as in Lemma~\ref{lemma:A_decomp_AL_AR}, and we can portion the domain $[\theta_i,\theta_{i+1}]$ into at most 9 cells and then run the Newton Raphson in every one of them to $\varepsilon$-approximate the extreme points of $f_i$.
\end{itemize}

For the running time to partition $\mathcal{P}$ from $C$ to compute and to sort $\{ CP_i\}_{I=1}^n$ 
takes $\mathcal{O}(n\log{n})$. Algorithm~\ref{Alg:Upper_Lower} is computing the upper and lower edges of each section which both take time $\mathcal{O}(n)$. Now in lines 4 - 7 of Algorithm~\ref{Alg:Max_Inter_Dir_theta}, the algorithm either finds or approximates the local maximum of each partition $LMR_i$. In case 2, the algorithm either runs at most 5 times Newton Raphson or in case 3 it runs as Theorem~\ref{Theor:Inter_Two_Sectors} states in time $\mathcal{O}(|b-a|~\log{\varepsilon}/\phi)$, where the length of the interval cannot be more than $\pi$ because $\mathcal{P}$ is a convex polygon. So the running time of the algorithm is $\mathcal{O}(n\log{n}+n\log{\varepsilon}/\phi)=$ $\mathcal{O}(n(\log{n}+\log{\varepsilon}/\phi))$.
\end{proof}

\noindent
{\bf Conclusion:}
The designed methods of finding the maximal intersection of a convex polygon with a rotating FOV directly can be applied to the special case of non-convex polygons where a rotating FOV could have only one component intersection and would not split the intersection into several disconnected parts. In this case, the presented methods still work because there is no restriction on gradients in the independent sub-problems, and the polygon can be decomposed to the already studied equivalent classes. On the other hand, the intersection of a non-convex polygon with a rotating FOV could create disconnected areas. However, the area functions are still applicable. The distinctive difference is that for every equivalence class, many intersection components may appear that lead to the calculation of the summation of multiple Left and Right functions. Finally, to complete the solution in this case, one must be aware of the other computational geometry problem of identifying these disconnected areas. 


\section{Technical Calculations and Proofs}
\label{sec:proof_calc}
In this section, we will provide proofs for Theorem~\ref{theor:Analytical_Area}, and Propostion~\ref{proposition:difcil_Ath}. Before proving Theorem~\ref{theor:Analytical_Area}, we will prove a formula that calculates the area of any convex polygon.

\begin{lemma}
\label{lemma:Shoelace}
The area of a polygon $\mathcal{P}=(P_1,\ldots,P_n)$, is given by the following formula:
\begin{align}
\label{eq:Polygon_Area}
poly\_area((x_1,y_1),\ldots,(x_n,y_n)) &=
\frac{1}{2} \cdot \left(
x_n ~y_1 + \sum_{i=1}^{n-1}{x_i~y_{i+1}} - 
x_1 y_n - \sum_{i=1}^{n-1}{x_{i+1}~y_i}
\right)
\nonumber
\\
&= 
\frac{1}{2} \left( det(P_n,P_1) + \sum_{i=1}^{n-1} det(P_i,P_{(i+1}) \right)
\end{align}
\end{lemma}

\begin{proof}
We will prove this formula with the use of induction on the number of vertices of a polygon.

\textbf{Base of the Induction:}
First, with the following claim, we prove the base of the induction for $n=3$.
The area of a triangle with coordinates $A(x_1,y_1)$,$B(x_2,y_2)$ and $C(x_3,y_3)$ is:
\begin{align*}
poly\_area(ABC) =
    \frac{1}{2} ( x_1y_2+x_2y_3+x_3y_1 - x_1y_3 - x_2y_1 - x_3y_2 )
\end{align*}
Indeed, by defining the two vectors $v_1 = B-A$ and $v_2=C-A$ it well known (see \cite{Axler1997}) that the area of the triangle $ABC$ is
\begin{align*}
    \frac{1}{2}det(v_1,v_2) = \frac{1}{2}
    \begin{tabular}{|c c|}
        $(x_2-x_1)$ & $(x_3-x_1)$ \\
        $(y_2-y_1)$ & $(y_3-y_1)$
    \end{tabular}
    =
    \frac{1}{2} ( x_1y_2+x_2y_3+x_3y_1 - x_1y_3 - x_2y_1 - x_3y_2 )
\end{align*}

\textbf{Induction Hypothesis:}
Let us assume that the polygon $P_1,\ldots,P_n$ with $n$ vertices in a counterclockwise order has an area that is given from:  
\begin{align*}
poly\_area(P_1,\ldots,P_n) &=
\frac{1}{2} \cdot det(P_i,P_{(i~mod~n)+1})
\end{align*}

\textbf{Induction Step:} We will prove that the above equation holds for the polygon $P_1,\ldots,P_{n+1}$ with $n+1$ vertices.
The area of the polygon $P_1,\ldots,P_{n+1}$ is the area of $P_1,\ldots,P_{n}$ plus the triangle's $P_nP_{n+1}P_1$ area, this means that

\begin{align*}
    poly\_area(P_1,\ldots,P_{n+1}) = 
    &\frac{1}{2} \cdot \left(
    x_n ~y_1 + \sum_{i=1}^{n-1}{x_i~y_{i+1}} - 
    x_1 y_n - \sum_{i=1}^{n-1}{x_{i+1}~y_i}
    \right) +
    \\
    +
    &\frac{1}{2} ( x_ny_{n+1}+x_{n+1}y_1+x_1y_n - x_ny_1 - x_{n+1}y_n - x_1y_{n+1} )
    \\
    =
    &\frac{1}{2} \cdot \left(
    x_{n+1} ~y_1 + \sum_{i=1}^{n}{x_i~y_{i+1}} - 
    x_1 y_{n+1} - \sum_{i=1}^{n}{x_{i+1}~y_i}
    \right)
\end{align*}
\end{proof}

\noindent
Proof of Theorem~\ref{theor:Analytical_Area}

\begin{proof}
The fact that one of the sectors has a fixed direction means that it can be defined using two intersecting lines $\varepsilon_1$ and $\varepsilon_2$. We will denote with $\varepsilon_1$ (resp. $\varepsilon_2$) the left (resp. the right) semi-line of $S(K,\theta_K,\phi_K)$. Notice that the slopes of $\varepsilon_1,\varepsilon_2$ are of angle $\theta_K$ and $\theta_K+\phi_K$ respectively. So we can define the semi-lines of $S(K,\theta_K,\phi_K)$ using a point, and a slope. The theorem is proven through the following lemma.

\begin{figure}[ht]
    \centering
    \scalebox{0.7}{
    \definecolor{yqqqqq}{rgb}{0.5019607843137255,0.,0.}
\definecolor{xdxdff}{rgb}{0.49019607843137253,0.49019607843137253,1.}
\definecolor{qqttcc}{rgb}{0.,0.2,0.8}
\definecolor{qqccqq}{rgb}{0.,0.8,0.}
\definecolor{zzttqq}{rgb}{0.6,0.2,0.}
\definecolor{qqzzff}{rgb}{0.,0.6,1.}
\definecolor{qqwuqq}{rgb}{0.,0.39215686274509803,0.}
\definecolor{uuuuuu}{rgb}{0.26666666666666666,0.26666666666666666,0.26666666666666666}
\definecolor{qqzzqq}{rgb}{0.,0.6,0.}
\definecolor{ududff}{rgb}{0.30196078431372547,0.30196078431372547,1.}
\begin{tikzpicture}[line cap=round,line join=round,>=triangle 45,x=1.0cm,y=1.0cm]
\clip(-11.7,-0.5) rectangle (0.5,8.);
\draw [shift={(-5.159393740790174,4.952143055386355)},line width=1.6pt,color=qqwuqq,fill=qqwuqq,fill opacity=0.10000000149011612] (0,0) -- (0.:0.49111569139188893) arc (0.:28.557705384660107:0.49111569139188893) -- cycle;
\draw [shift={(-5.159393740790174,0.21531696088107866)},line width=2.pt,color=zzttqq,fill=zzttqq,fill opacity=0.10000000149011612] (0,0) -- (56.71501006225292:0.49111569139188893) arc (56.71501006225292:103.78711748401521:0.49111569139188893) -- cycle;
\draw [shift={(-5.159393740790174,0.21531696088107866)},line width=2.pt,color=qqttcc,fill=qqttcc,fill opacity=0.10000000149011612] (0,0) -- (0.:0.49111569139188893) arc (0.:56.71501006225292:0.49111569139188893) -- cycle;
\fill[line width=2.pt,color=zzttqq,fill=zzttqq,fill opacity=0.10000000149011612] (-5.803401263375599,2.839793156357704) -- (-3.0039124901504244,3.4985976737639572) -- (-0.3207801583523553,7.585609162493171) -- (-6.184794752758505,4.394057817703224) -- cycle;
\draw [shift={(-1.9572383529297288,3.744911806704338)},line width=2.pt,color=qqwuqq,fill=qqwuqq,fill opacity=0.10000000149011612] (0,0) -- (0.:0.49111569139188893) arc (0.:13.242483943385635:0.49111569139188893) -- cycle;
\draw [line width=2.pt,domain=-11.7:0.5] plot(\x,{(--76.79914092070155--5.386298304762439*\x)/9.8965451144985});
\draw [line width=2.pt,domain=-11.7:0.5] plot(\x,{(-27.584748764628948-1.5435769155688561*\x)/-6.559193404983679});
\draw [line width=0.4pt,color=qqzzqq] (-5.159393740790174,-0.5) -- (-5.159393740790174,8.);
\draw [shift={(-5.159393740790174,4.952143055386355)},line width=1.6pt,color=qqwuqq] (0.:0.49111569139188893) arc (0.:28.557705384660107:0.49111569139188893);
\draw [shift={(-5.159393740790174,4.952143055386355)},line width=1.6pt,color=qqwuqq] (0.:0.3928925531135111) arc (0.:28.557705384660107:0.3928925531135111);
\draw [line width=2.pt,color=qqzzff,domain=-11.7:0.5] plot(\x,{(--39.06807502424048--7.370292201612092*\x)/4.8386135824378185});
\draw [line width=2.pt,color=qqzzff,domain=-11.7:0.5] plot(\x,{(--56.19305667643323--11.004096104127795*\x)/-2.7002419311424806});
\draw [line width=0.4pt,color=qqccqq,domain=-11.7:0.5] plot(\x,{(--0.21531696088107866-0.*\x)/1.});
\draw [line width=2.pt,color=qqwuqq] (-0.8880080047409551,3.744911806704338)-- (-1.9572383529297288,3.744911806704338);
\draw [line width=2.pt,color=qqwuqq] (-5.159393740790174,4.952143055386355)-- (-4.303766506544882,4.952143055386355);
\draw [line width=0.4pt,color=yqqqqq,domain=-11.7:0.5] plot(\x,{(--5.249662424407623--1.282377421212487*\x)/-6.347050461489317});
\draw [line width=0.4pt,color=yqqqqq,domain=-11.7:0.5] plot(\x,{(--3.0233710280431114--0.5442604709467226*\x)/1.});
\begin{large}
\draw [fill=ududff] (-1.9572383529297288,3.744911806704338) circle (2.5pt);
\draw[color=ududff] (-2.0326238389284805,3.4321247394979766) node {$B$};
\draw [fill=ududff] (-5.159393740790174,0.21531696088107866) circle (2.5pt);
\draw[color=ududff] (-5.977919893109988,-0.005685100245273106) node {$C$};
\draw [fill=uuuuuu] (-5.159393740790174,4.952143055386355) circle (2.0pt);
\draw[color=uuuuuu] (-5.588581063439721,5.331105412879962) node {$E'$};
\draw [fill=uuuuuu] (-5.159393740790174,2.991347636297303) circle (2.0pt);
\draw[color=uuuuuu] (-4.7973896185374735,3.41575421645158) node {$E$};
\draw[color=qqwuqq] (-4.226273927145584,5.183770705462393) node {$\omega$};
\draw [fill=uuuuuu] (-0.3207801583523553,7.585609162493171) circle (2.5pt);
\draw[color=uuuuuu] (-0.02723476574493397,7.385606055202713) node {$P_2$};
\draw[color=zzttqq] (-4.962947464233419,1.0420283747240981) node {$\phi$};
\draw[color=qqttcc] (-4.308126542377566,0.5509126833322053) node {$\theta$};
\draw [fill=uuuuuu] (-6.184794752758505,4.394057817703224) circle (2.0pt);
\draw[color=uuuuuu] (-6.755519737813813,4.782692890825682) node {$P_3$};
\draw [fill=uuuuuu] (-5.803401263375599,2.839793156357704) circle (2.0pt);
\draw[color=uuuuuu] (-6.21173875383949,2.272672279562164) node {$P_4$};
\draw [fill=uuuuuu] (-3.0039124901504244,3.4985976737639572) circle (2.0pt);
\draw[color=uuuuuu] (-2.87570577581789,3.0474174479076606) node {$P_1$};
\draw [fill=uuuuuu] (-11.506444202279491,1.4976943820935658) circle (2.0pt);
\draw[color=uuuuuu] (-11.593009298023919,2.0406302805542804) node {$K$};
\draw[color=qqwuqq] (-1.5087671014437989,4.316132984003383) node {$\beta$};
\draw [fill=xdxdff] (-2.257669971942686,6.531436600376588) circle (2.5pt);
\draw[color=xdxdff] (-2.1472175002532543,6.837193533148433) node {$D$};
\draw[color=yqqqqq] (-8.5,0.6) node {$\varepsilon_{\theta_{max}}$};
\draw[color=yqqqqq] (-0.7,2) node {$\varepsilon_{\theta_{min}}$};
\draw[color=qqzzff] (-7.2035189370316445,7.413981125838406) node {$\varepsilon_\ell$};
\draw[color=qqzzff] (-1,6) node {$\varepsilon_r$};
\draw[color=uuuuuu] (-8.5,3.5) node {$\varepsilon_1$};
\draw[color=uuuuuu] (-8.5,1.8) node {$\varepsilon_2$};
\draw[color=zzttqq] (-3.9209044755993984,4.349722659347877) node {$A(\theta,\phi))$};
\draw[color=qqwuqq] (-4.6,7.5) node {$\varepsilon_y$};
\end{large}
\end{tikzpicture}
    }
    \caption{The area of intersection of two sectors. There are the two lines $\varepsilon_1$ and $\varepsilon_2$, and the sector $S(C,\theta,\phi)$ with the blue lines. The positive angles $\omega$ and $\beta$ correspond to the slopes of $\varepsilon_1$ and $\varepsilon_2$ respectively. Function $A(\theta,\phi) $ is defined between the lines $\varepsilon_{\theta_{max}}$, and $\varepsilon_{\theta_{min}}$ which correspond to angles $\omega$ and $\hat{CK}$, where $\theta \in \left(\omega,\hat{CK}\right)$.Finally, $E' = \varepsilon_y \cap \varepsilon_1$, $E = \varepsilon_y \cap \varepsilon_2$ and $d_1=|CE'|^2$, $d_2=|CE|^2$.}
    \label{fig:General_Case22}
\end{figure}
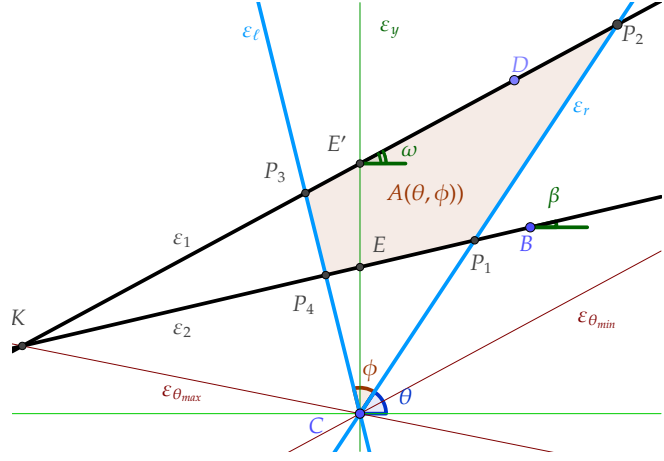

\begin{lemma}
\label{lemma:General_Case2}
Let $B=(x_b,y_b)$ and $D=(x_d,y_d)$ be two points on the plane, $\beta \leq \omega\in (-\pi/2,\pi/2)$ two positive angles, and a sector $S(C,\theta,\phi)$ with center $C=(x_0,y_0)$ and inner angle $\phi \in (0,\pi)$. The line $\varepsilon_1$ is defined by the point $D$, and the slope $\tan{\omega}$; and $\varepsilon_2$ from $B$, and $\tan{\beta}$ respectively.
The quadrilateral's area created by $\varepsilon_1 \cap \varepsilon_2 \cap S$, is given from the following function:
\begin{align}
\label{eq:General_Case_Alt_lemma}
    &A(\theta,\phi) = 
    \frac{d_1 \sin{\phi} ~\cos^2{\omega} }{2\sin{(\theta+\phi-\omega)}\sin{(\theta-\omega)}} + \frac{d_2 \sin{\phi} ~\cos^2{\beta} }{2\sin{(\theta+\phi-\beta)}\sin{(\theta-\beta)}} 
    && (\theta, \phi) \in \mathcal{R}
\end{align}
where $\mathcal{R} = (\theta_{min},\theta_{max}-\phi) \times (0,\pi)$, $K=(x_k,y_k)$ is the point of $\varepsilon_1\cap\varepsilon_2$
\begin{align*}
    \theta_{min}= \begin{cases}
        \omega \hspace{1.3cm}
        \mbox{ if } \frac{x_0-x_k}{|x_0-x_k|}>0
        \\
        \hat{CK} \hspace{1.1cm}
        \mbox{ if } \frac{x_0-x_k}{|x_0-x_k|}<0
    \end{cases}
    \hspace{2cm}
    \theta_{max}= \begin{cases}
        \hat{CK} \hspace{0.75cm}
        \mbox{ if } \frac{x_0-x_k}{|x_0-x_k|}>0
        \\
        \beta \hspace{1.1cm} 
        \mbox{ if } \frac{x_0-x_k}{|x_0-x_k|}<0
    \end{cases}
\end{align*}
\begin{align*}
    d_1&= \frac{x_0-x_k}{|x_0-x_k|} (\tan{\omega}(x_0-x_d)+y_d-y_0)^2,
    &d_2= \frac{x_k-x_0}{|x_k-x_0|} (\tan{\beta}(x_0-x_b)+y_b-y_0)^2 
\end{align*}
\end{lemma}

\noindent
\textit{Proof of Lemma~\ref{lemma:General_Case2}}: 
Without loss of generality, we will prove the lemma in the case where $\frac{x_0-x_k}{|x_0-x_k|}>0$.
The equations of the lines $\varepsilon_1$, $\varepsilon_2$, $\varepsilon_\ell$, and $\varepsilon_r$ (see Figure~\ref{fig:General_Case22}) are:
\begin{align*}
    &\varepsilon_1:\ y = \tan{\omega}~(x-x_d)+y_d &&
    &\varepsilon_2:\ y = \tan{\beta}~(x-x_b)+y_b
    \\
    &\varepsilon_r:\ y = \tan{\theta}~(x-x_0)+y_0 &&
    &\varepsilon_\ell:\ y = \tan{(\theta+\phi)}~(x-x_0)+y_0 
\end{align*}
Line $\varepsilon_2$ passes through the point $E = (x_0,y_e)$ where $y_e$ can be expressed as $y_e = y_0+d_1$, respectively $\varepsilon_1$ passes through the point 
$E' = (x_0,y'_e)$, where $y'_e = y_0+d_1+d_2$, $d_1,d_2>0$.
\begin{align*}
    \varepsilon_1:\ y = \tan{\omega}~(x-x_0)+y_0+d_1+d_2 &&
    \varepsilon_2:\ y = \tan{\beta}~(x-x_0)+y_0+d_1
\end{align*}
We begin by computing the coordinates of the points $P_1,P_2,P_3,P_4$ as the intersections of the lines $(\varepsilon_r \cap \varepsilon_2)$, $(\varepsilon_r \cap \varepsilon_1)$, $(\varepsilon_\ell \cap \varepsilon_1)$ and $(\varepsilon_\ell \cap \varepsilon_2)$ respectively.
\begin{align*}
    &(\varepsilon_r \cap \varepsilon_2):\tan{\theta}~(x-x_0)+y_0 = \tan{\beta}~(x-x_0)+y_0+d_1 \Rightarrow (\tan{\theta}-\tan{\beta})~(x-x_0)=d_1
    \\
    &
    \Rightarrow
    x_1 = x_0+\frac{d_1}{\tan{\theta}-\tan{\beta}},
    \hspace{4.3cm}
    y_1 = y_0+\frac{d_1\tan{\theta}}{\tan{\theta}-\tan{\beta}}
    \\
    &(\varepsilon_r \cap \varepsilon_1):\tan{\theta}~(x-x_0)+y_0 = \tan{\omega}~(x-x_0)+y_0+d_1+d_2 \Rightarrow
    \\
    &\Rightarrow
    (\tan{\theta}-\tan{\omega})~(x-x_0)=d_1+d_2\Rightarrow \\
    &\Rightarrow x_2 = x_0+\frac{d_1+d_2}{\tan{\theta}-\tan{\omega}}, \hspace{+4.1cm}  y_2 = y_0+\frac{(d_1+d_2)\tan{\theta}}{\tan{\theta}-\tan{\omega}}
    \\
    &(\varepsilon_\ell \cap \varepsilon_1):\tan{(\theta+\phi)}~(x-x_0)+y_0 = \tan{\omega}~(x-x_0)+y_0+d_1+d_2 \Rightarrow
    \\
    &\Rightarrow
    (\tan{(\theta+\phi)}-\tan{\omega})~(x-x_0)=d_1+d_2\Rightarrow \\
    &\Rightarrow x_3 = x_0+\frac{d_1+d_2}{\tan{(\theta+\phi)}-\tan{\omega}}, \hspace{+2.2cm}  y_3 = y_0+\frac{(d_1+d_2)\tan{(\theta+\phi)}}{\tan{(\theta+\phi)}-\tan{\omega}}
\\
    &(\varepsilon_\ell \cap \varepsilon_2):\tan{(\theta+\phi)}~(x-x_0)+y_0 = \tan{\beta}(x-x_0) + y_0 + d_1 \Rightarrow 
    \\
    &\Rightarrow
    (\tan{(\theta+\phi)}-\tan{\beta})(x-x_0) = d_1
    \\
    &\Rightarrow    
    x_4 = x_0+\frac{d_1}{\tan{(\theta+\phi)-\tan{\beta}}}, \hspace{2.4cm}
    y_4 = y_0+\frac{d_1\tan{(\theta+\phi)}}{\tan{(\theta+\phi)-\tan{\beta}}}
    \\
    &(\varepsilon_1 \cap \varepsilon_2):\tan{\omega}~(x-x_0)+y_0+d_1+d_2 = \tan{\beta}(x-x_0) + y_0 + d_1 \Rightarrow \\
    &\Rightarrow 
    x_{int} = x_0+\frac{d_2}{\tan{\beta}-\tan{\omega}}, \hspace{2.8cm}
    y_{int} = y_0+d_1+ \frac{d_2\tan{\beta}}{\tan{\beta}-\tan{\omega}}
\end{align*}

\noindent
Now we can use the shoelace formula that computes the area of a polygon with vertices, ordered counterclockwise, $P_1\ldots P_n$, from Lemma~\ref{lemma:Shoelace}. 
\begin{align}
    poly\_area(P_1P_2\ldots P_n) = \frac{1}{2} \left(
    det(P_n,P_1)+
    \sum_{i=1}^{n}det(P_iP_{i+1})
    \right)
\end{align}
to compute the area of the quadrilateral $P_1 P_2 P_3 P_4$
\begin{align}
\label{eq:2_poly_area_proof}
    2\cdot poly&\_area((x_1,y_1),\ldots,(x_4,y_4)) = 
    x_1y_2+x_2y_3+x_3y_4+x_4y_1 - y_1x_2-y_2x_3-y_3x_4-y_4x_1 = \nonumber
    \\
    =& y_1(x_4-x_2)+ y_2(x_1-x_3)+ y_3(x_2-x_4)+y_4(x_3-x_1)
\end{align}

    $
   \bullet \hspace{0.6cm} y_1(x_4-x_2) = y_1
    \left( x_0+\frac{d_1}{\tan{(\theta+\phi)}-\tan{\beta}} -
    x_0-\frac{d_1+d_2}{\tan{\theta}-\tan{\omega}}
    \right)=
    \left(
    y_0+\frac{d_1\tan{\theta}}{\tan{\theta}-\tan{\beta}}
    \right)
    \left( \frac{d_1}{\tan{(\theta+\phi)-\tan{\beta}}} -
    \frac{d_1+d_2}{\tan{\theta}-\tan{\omega}}
    \right)
    $

    $
    \bullet \hspace{0.6cm} y_2(x_1-x_3) = y_2
    \left( x_0+\frac{d_1}{\tan{\theta}-\tan{\beta}} -
    x_0-\frac{d_1+d_2}{\tan{(\theta+\phi)}-\tan{\omega}}    
    \right)=%
    \left(
    y_0+\frac{(d_1+d_2)\tan{\theta}}{\tan{\theta}-\tan{\omega}}
    \right)
    \left( \frac{d_1}{\tan{\theta}-\tan{\beta}} -
    \frac{d_1+d_2}{\tan{(\theta+\phi)}-\tan{\omega}}
    \right)
    $
    
    $\bullet \hspace{0.6cm} y_3(x_2-x_4) = y_3
    \left( x_0+\frac{d_1+d_2}{\tan{\theta}-\tan{\omega}} -
    x_0-\frac{d_1}{\tan{(\theta+\phi)}-\tan{\beta}}
    \right)=
    \left(
    y_0+\frac{(d_1+d_2)\tan{(\theta+\phi)}}{\tan{(\theta+\phi)}-\tan{\omega}}
    \right)
    \left( \frac{d_1+d_2}{\tan{\theta}-\tan{\omega}} -
    \frac{d_1}{\tan{(\theta+\phi)}-\tan{\beta}}
    \right)$

    $\bullet \hspace{0.6cm} y_4(x_3-x_1) = y_4
    \left( 
    x_0+\frac{d_1+d_2}{\tan{(\theta+\phi)}-\tan{\omega}}-
    x_0-\frac{d_1}{\tan{\theta}-\tan{\beta}}
    \right)=
    \left(
    y_0+\frac{d_1\tan{(\theta+\phi)}}{\tan{(\theta+\phi)}-\tan{\beta}}
    \right)
    \left( 
    \frac{d_1+d_2}{\tan{(\theta+\phi)}-\tan{\omega}}-
    -\frac{d_1}{\tan{\theta}-\tan{\beta}}
    \right)$

$    \bullet \hspace{0.6cm} y_1(x_4-x_2) + y_3(x_2-x_4) = (y_3-y_1)(x_2-x_4) = 
    \left(
    \frac{(d_1+d_2)\tan{(\theta+\phi)}}{\tan{(\theta+\phi)}-\tan{\omega}} -
    \frac{d_1\tan{\theta}}{\tan{\theta}-\tan{\beta}}
    \right)
    \left( \frac{d_1+d_2}{\tan{\theta}-\tan{\omega}} -
    \frac{d_1}{\tan{(\theta+\phi)}-\tan{\beta}}
    \right)=$

$    \hspace{0.9cm} =
    \frac{(d_1+d_2)^2 \tan{(\theta+\phi)}}{(\tan{(\theta+\phi)}-\tan{\omega})(\tan{\theta}-\tan{\omega})} -
    \frac{d_1(d_1+d_2)\tan{(\theta+\phi)}}{(\tan{(\theta+\phi)}-\tan{\omega})(\tan{(\theta+\phi)}-\tan{\beta})} +
    \frac{d_1^2 \tan{\theta}}{(\tan{(\theta+\phi)}-\tan{\beta})(\tan{\theta}-\tan{\beta})} -
    \frac{d_1(d_1+d_2)\tan{\theta}}{(\tan{\theta}-\tan{\omega})(\tan{\theta}-\tan{\beta})}
$

$    \bullet \hspace{0.3cm} y_2(x_1-x_3) + y_4(x_3-x_1) = (y_4-y_2)(x_3-x_1) = 
    \left(
    \frac{d_1\tan{(\theta+\phi)}}{\tan{(\theta+\phi)}-\tan{\beta}} -
    \frac{(d_1+d_2)\tan{\theta}}{\tan{\theta}-\tan{\omega}}
    \right)
    \left( 
    \frac{d_1+d_2}{\tan{(\theta+\phi)}-\tan{\omega}}
    -\frac{d_1}{\tan{\theta}-\tan{\beta}}
    \right)=$

$    \hspace{0.9cm}=
    -\frac{(d_1+d_2)^2 \tan{\theta}}{(\tan{(\theta+\phi)}-\tan{\omega})(\tan{\theta}-\tan{\omega})} - 
    \frac{d_1^2 \tan{(\theta+\phi)}}{(\tan{(\theta+\phi)}-\tan{\beta})(\tan{\theta}-\tan{\beta})}+
    \frac{d_1(d_1+d_2)\tan{(\theta+\phi)}}{(\tan{(\theta+\phi)}-\tan{\omega})(\tan{(\theta+\phi)}-\tan{\beta})} + 
    \frac{d_1(d_1+d_2)\tan{\theta}}{(\tan{\theta}-\tan{\omega})(\tan{\theta}-\tan{\beta})}$
\\[1.5ex]
\noindent
By plugging our calculation to equation~(\ref{eq:2_poly_area_proof}):
\begin{align*}
    2\cdot poly\_area((x_1,y_1),\ldots,(x_4,y_4)) = (y_3-y_1)(x_2-x_4) +(y_4-y_2)(x_3-x_1)
     \nonumber
\end{align*}
\begin{align*}
    \Rightarrow poly\_area = 
     \frac{(d_1+d_2)^2 (\tan{(\theta+\phi)}-\tan{\theta})}{2(\tan{(\theta+\phi)}-\tan{\omega})(\tan{\theta}-\tan{\omega})} + \frac{d_1^2 (\tan{\theta}-\tan{(\theta+\phi)})}{2(\tan{(\theta+\phi)}-\tan{\beta})(\tan{\theta}-\tan{\beta})}
\end{align*}

\noindent
If we express the lines $\varepsilon_1$ and $\varepsilon_2$ using the points $B=(x_b,y_b)$  and $D=(x_d,y_d)$ respectively then if we take into account that $E\in \varepsilon_2$ and $E'\in \varepsilon_1$ then
\begin{align*}
    y_e = (x_0-x_b)\tan{\beta}+y_b &\Rightarrow d_1 = (x_0-x_b)\tan{\beta}+y_b-y_0 \\
    y'_e = (x_0-x_d)\tan{\omega}+y_d &\Rightarrow d_1+d_2 = (x_0-x_d)\tan{\omega}+y_d-y_0
\end{align*}
Thus obtaining the equation
\small
\begin{align}
\label{eq:General_Case_2}
    A(C,B,D,\beta,\omega,\theta,\phi)
    =&
    \frac{(\tan{\omega}(x_0-x_d)+y_d-y_0)^2 (\tan{(\theta+\phi)}-\tan{\theta})}{2(\tan{(\theta+\phi)}-\tan{\omega})(\tan{\theta}-\tan{\omega})} + 
    \nonumber
    \\
    &+
    \frac{(\tan{\beta}(x_0-x_b)+y_b-y_0)^2 (\tan{\theta}-\tan{(\theta+\phi)})}{2(\tan{(\theta+\phi)}-\tan{\beta})(\tan{\theta}-\tan{\beta})}
\end{align}
\normalsize
Now note the following equations
\begin{align}
    \tan{a}-\tan{b} = \frac{\sin{a}}{\cos{a}} -\frac{\sin{b}}{\cos{b}} = \frac{\sin{(a-b)}}{\cos{a}\cos{b}}
\end{align}
\begin{align}
    \frac{\tan{a}-\tan{b}}{(\tan{a}-\tan{c})(\tan{b}-\tan{c})} &= \frac{\sin{(a-b)}\cos{a}\cos{b}\cos^2{c}}{\sin{(a-c)}\sin{(b-c)}\cos{a}\cos{b}} =
    \frac{\sin{(a-b)}\cos^2{c}}{\sin{(a-c)}\sin{(b-c)}} 
\end{align}
By substituting appropriately the above equation to equation~\ref{eq:General_Case_2} we have
\small{
\begin{align*}
    &A(C,B,D,\beta,\omega,\theta,\phi) = 
    \frac{(\tan{\omega}(x_0-x_d)+y_d-y_0)^2 (\tan{(\theta+\phi)}-\tan{\theta})}{2(\tan{(\theta+\phi)}-\tan{\omega})(\tan{\theta}-\tan{\omega})} +
    \\
    &\hspace{3.3cm}
    + \frac{(\tan{\beta}(x_0-x_b)+y_b-y_0)^2 (\tan{\theta}-\tan{(\theta+\phi)})}{2(\tan{(\theta+\phi)}-\tan{\beta})(\tan{\theta}-\tan{\beta})}
    \\
    &=
    \frac{(\tan{\omega}(x_0-x_d)+y_d-y_0)^2 \sin{\phi}\cos^2{\omega}}{2\sin{(\theta+\phi-\omega)}\sin{(\theta-\omega)}} - \frac{(\tan{\beta}(x_0-x_b)+y_b-y_0)^2 \sin{\phi}\cos^2{\beta}}{2\sin{(\theta+\phi-\beta)}\sin{(\theta-\beta)}}
\end{align*}
}
\normalsize
\hspace{9cm}
\textit{End of proof of Lemma~\ref{lemma:General_Case2} }
\qed

By substituting $\omega=\theta_K+\phi_K$, and $\beta=\theta_K$ at~(\ref{eq:General_Case_Alt_lemma}) we have equation~(\ref{eq:General_Case_Alt}).

\textit{End of proof of Theorem~\ref{theor:Analytical_Area}}
\end{proof}

\noindent
Here we provide a proof for Proposition~\ref{proposition:difcil_Ath}.

\begin{proof}
    We set $\omega=\theta_K+\phi_K$, and $\beta = \theta_K$, so equation~(\ref{eq:General_Case_Alt}) is
    \begin{align*}
    &A_\phi(\theta) = 
    \frac{d_1 \sin{\phi}~\cos^2{\omega} }{2\sin{(\theta+\phi-\omega)}\sin{(\theta-\omega)}} - \frac{d_2 \sin{\phi}~\cos^2{\beta}}{2\sin{(\theta+\phi-\beta)}\sin{(\theta-\beta)}}
    \end{align*}
    Now to show that function $A$ can be expressed as a polynomial function we will use the determinant to show the complexity of the polygons by avoiding the calculations. Notice that
\begin{align*}
    \sin(\theta+\phi-\omega) = 
    \begin{vmatrix}
        \cos{\omega} & \cos{\theta} & \sin{\theta}
        \\
        \sin{\omega} & \sin{\theta} & -\cos{\theta}
        \\
        0 & \sin{\phi} & \cos{\phi}
    \end{vmatrix}
    &&
    \sin{(\theta-\omega)} = 
    \begin{vmatrix}
        \sin{\theta} & \cos{\theta}
        \\
        \sin{\omega} & \cos{\omega}
    \end{vmatrix}
\end{align*}

\begin{align}
\label{eq:det_sin}
    \sin(\theta+\phi-\omega)\sin{(\theta-\omega)} =
    \begin{vmatrix}
        \cos{\omega} & \cos{\theta} & \sin{\theta} & 0 & 0 
        \\
        \sin{\omega} & \sin{\theta} & -\cos{\theta} & 0 & 0
        \\
        0 & \sin{\phi} & \cos{\phi} & 0 & 0 
        \\
        0 & 0 & 0 & \sin{\theta} & \cos{\theta}
        \\
        0 & 0 & 0 & \sin{\omega} & \cos{\omega}
    \end{vmatrix}
\end{align}

Let $x=\sin{\theta}$ and  $\cos{\theta}=\sqrt{1-x^2}$ then equation~(\ref{eq:det_sin}) is
\begin{align*}
    \sin(\theta+\phi-\omega)\sin{(\theta-\omega)} =
    \begin{vmatrix}
        \cos{\omega} & \sqrt{1-x^2} & x & 0 & 0 
        \\
        \sin{\omega} & x & -\sqrt{1-x^2} & 0 & 0
        \\
        0 & \sin{\phi} & \cos{\phi} & 0 & 0 
        \\
        0 & 0 & 0 & x & \sqrt{1-x^2}
        \\
        0 & 0 & 0 & \sin{\omega} & \cos{\omega}
    \end{vmatrix}
\end{align*}
The determinant of the above equation will produce an equation of the following form
\begin{align*}
    P_1(x) + P_2(x)\sqrt{1-x^2}
\end{align*}
Where $P_1$ is a polynomial of order at most 2 and $P_2$ is a polynomial of order 1.
If we apply the same technique for $\sin(\theta+\phi-\beta)\sin{(\theta-\beta)}$ we get the rational form of $A_\phi(\theta)$ 
\begin{align*}
    A_\phi(\theta) &= \frac{D_1}{P_1(x) + P_2(x)\sqrt{1-x^2}} + \frac{D_2}{P_3(x) + P_4(x)\sqrt{1-x^2}} 
    \\
    &= \frac{D_1P_3(x) + D_1P_4(x)\sqrt{1-x^2} + D_2P_1(x) + D_2P_2(x)\sqrt{1-x^2} }{(P_1(x) + P_2(x)\sqrt{1-x^2})(P_3(x) + P_4(x)\sqrt{1-x^2})}
\end{align*}
\end{proof}

\bibliographystyle{abbrv}
\bibliography{biblio}

\end{document}